\documentclass[12pt]{article}

\usepackage{graphicx}
\usepackage{amsmath}
\usepackage{amssymb}
\usepackage{amsthm}
\newcommand{\be}{\begin{equation}}
\newcommand{\ee}{\end{equation}}
\newcommand{\ba}{\begin{array}}
\newcommand{\ea}{\end{array}}
\newcommand{\bea}{\begin{eqnarray}}
\newcommand{\eea}{\end{eqnarray}}

\newcommand{\calC}{{\cal C }}

\newcommand{\calL}{{\cal L }}

\newcommand{\calT}{{\cal T }}

\newcommand{\calP}{{\cal P }}

\newcommand{\calS}{{\cal S }}
\newcommand{\calG}{{\cal G }}

\newcommand{\calB}{{\cal B }}
\newcommand{\ZZ}{\mathbb{Z}}

\newcommand{\modtwo}{\; (\mathrm{mod}\;  2)}
\newcommand{\modn}[1]{\; (\mathrm{mod}\; {#1})}
\newcommand{\hexlat}{\Omega_{\rm hex}}
\newcommand{\la}{\langle}
\newcommand{\ra}{\rangle}

\newcommand{\nn}{\nonumber}

\newcommand{\trace}{\mathop{\mathrm{Tr}}\nolimits}

\newtheorem{dfn}{Definition}
\newtheorem{lemma}{Lemma}
\newtheorem{prop}{Proposition}
\newtheorem{theorem}{Theorem}

\title{Topological order in an exactly solvable 3D spin model}

\author{Sergey Bravyi\footnote{IBM Watson Research Center, Yorktown Heights, NY 10598}, \;
Bernhard Leemhuis\footnote{Institute for Theoretical Physics, University of Amsterdam, Valckenierstraat 65,
1018 XE Amsterdam}, \; and \;
Barbara M. Terhal${}^*$
}

\date{\today}

\begin{document}

\maketitle

\begin{abstract}
We study a 3D generalization of the toric code model introduced recently by Chamon.
This is an exactly solvable spin model with six-qubit nearest-neighbor interactions on an FCC lattice whose ground space exhibits topological quantum order.  The elementary excitations of this model which we call monopoles can be geometrically described as the corners of rectangular-shaped membranes. We prove that the creation of an isolated monopole separated from other monopoles by a distance $R$
requires an operator acting on $\Omega(R^2)$ qubits. Composite particles that consist of two monopoles (dipoles) and four monopoles (quadrupoles) can be described as end-points of strings. The peculiar feature of the model is that dipole-type strings are rigid, that is, such strings must be aligned with face-diagonals of the lattice.  For periodic boundary conditions the ground space can encode $4g$ qubits where $g$ is the greatest common divisor of the lattice dimensions.
We describe a complete set of  logical operators acting on the encoded qubits in terms of  closed strings and closed membranes.
 \end{abstract}

\newpage

\tableofcontents

\newpage

\section{Introduction}

In the last decade, many two-dimensional spin models which exhibit topological order have been proposed and studied. Well-known examples include the 2D toric code and quantum doubles of finite groups~\cite{kitaev:anyons},
Levin-Wen string-net models~\cite{LW:stringnet}, topological color codes~\cite{BMD:topo},
Kitaev's honeycomb model~\cite{kitaev:anyon_pert},  quantum dimer models
on the triangular lattice~\cite{MS01,IF02}, and quantum loop gases~\cite{FF05}.
The corresponding topological phases can be described by a suitable class of anyons~\cite{kitaev:anyon_pert}
which captures many essential features of a model including the superselection sectors, the ground state degeneracy, and the behavior of excitations under braiding and fusion.

Much less is known about topological order in three dimensions and a general classification of topological phases in 3D largely remains an open problem.
As a first step in this direction, a 3D generalization of 2D toric codes has been studied by several authors~\cite{HZW:3dmodels,CC:toric, Wegner71}.
The 3D toric code can be viewed as a 3D lattice $\ZZ_2$ gauge theory.
In this theory electric charges are point-like excitations that can be created by string-like operators. On the other hand, magnetic excitations correspond to closed loops of $\ZZ_2$ flux. Such excitations can be created by membrane-like operators.
The 3D toric code features a macroscopic energy barrier for the logical membrane-like $\bar{X}$ operator, but not for the string-like logical $\bar{Z}$ operator. This implies that one encode a classical bit in such system whose value will be protected at finite temperature~\cite{CC:toric}.

Other 3D models with both string and membrane-like logical operators have been analyzed in~\cite{HZW:3dmodels,Bombin07}. In these models membrane-like and string-like operators always appear in pairs, such that the two operators anticommute if the string crosses the membrane at an odd number of points. Hence, one may be led to believe that some kind of duality must be at play between the supports of the logical operators.
Such duality would say, for example, that if one logical operator of a qubit, say $\bar{X}$, is membrane-like, then the other logical operator, say, $\bar{Z}$,  must be string-like or even point-like. Such duality would first of all imply that the distance of 3D error-correcting codes would be bounded by the linear size of the system. Secondly, it would provide strong evidence for a no-go result concerning thermal stability of a 3D passive quantum memory~\cite{BT:mem}.

In this paper we analyze a spin model originally introduced by Chamon~\cite{Chamon05}
in which some of these beliefs can be examined with a greater care.
The elementary excitations of the model are point-like particles which we call monopoles. Monopoles carry a non-trivial topological charge meaning that no local operator can create an isolated monopole from the ground space.  In contrast to all previously studied 3D topological spin models, monopoles cannot be created by string-like operators. Rather, monopoles can be described as the corners of rectangular-shaped membranes. More specifically,  suppose an excited state contains an isolated monopole separated from other monopoles by a distance $R$. We prove that creating such a state starting from the ground space requires
 an operator acting on $\Omega(R^2)$ qubits.
 It implies that local errors cannot induce diffusion in a dilute gas
 of monopoles since
 only highly non-local operators can move monopoles between adjacent sites~\cite{Chamon05}.
 Composite particles that consist of two monopoles (dipoles) and four monopoles (quadrupoles) can be described as end-points of strings. The peculiar feature of the model is that dipole-type strings are rigid, that is, these strings must be aligned with face-diagonals of the lattice. The quadrupole-type strings are partially flexible, meaning that such strings can only be deformed within planes orthogonal to body-diagonals of the lattice.

The degenerate ground space of the model defines a quantum error correcting code with
$k=4g$ logical qubits, where $g$ is the greatest common divisor of the lattice dimensions.
We describe a complete set of logical operators in terms of closed strings and closed membranes
 for the special case $g=1$. The rigidity of dipole-type strings in our model leads to very interesting features of the corresponding quantum code. Specifically, if the lattice has periodic boundary conditions and the lattice dimensions are pairwise co-prime, a rigid string aligned with a face-diagonal must fully fill up a two-dimensional plane before it gets closed.  Hence closed rigid strings can become membrane-like objects. We propose a subsystem encoding of a single logical qubit in which the only relevant
 logical  operators are those associated with closed rigid strings and closed membranes. In this encoding
closed flexible (quadrupole-type) strings can only affect the gauge subsystem.
Several strategies of minimizing the weight of logical operators composed of rigid strings are discussed.

It is important to emphasize the essential difference between membranes describing  logical operators in our code and the ones in the 4D toric code~\cite{Dennis:2001}. In our model membranes must have a rectangular shape and the membrane operator creates excitations only near the corners of the rectangle. Such membranes can be increased in size without paying extra energy penalty.
In contrast, membranes in the 4D toric code may take arbitrary shapes
and the membrane operator creates excitations along the entire boundary of the membrane.
Increasing the size of such membranes costs energy growing linearly with the length of the boundary.
The presence of this ``energy barrier" is responsible for the thermal stability of 4D toric code~\cite{AHHH:4d} manifesting itself in the exponentially large (as a function of  the lattice size) relaxation time of the error-corrected logical operators.
Although our 3D model does not feature a linearly growing energy barrier similar to the 4D toric code, its dynamics towards a thermal equilibrium might be characterized by glassiness in the creation/annihilation of isolated monopoles as was argued in~\cite{Chamon05}. Such glassiness might result from the fact that local errors cannot induce diffusion of isolated monopoles.
We do not expect that the relaxation time of the error-corrected logical operators in our 3D model
will grow with the lattice size, see the discussion in Section~\ref{subs:rclosed}.

Our main technical results pertaining to the Chamon's model~\cite{Chamon05}  which we introduce formally in the next section can be summarized as follows.
\begin{itemize}
\item Complete classification of bulk excitations
\item Lower bound on the weight of operators creating isolated monopoles
\item General formula for the ground space degeneracy on a $3$-torus
\item Subsystem encoding of a qubit using only closed rigid strings
\item Proof of the zero-temperature stability
\end{itemize}

\subsection{The spin model}
We consider a 3D cubic lattice $\Lambda$ with
periodic boundary conditions and linear dimensions $L_x,L_y,L_z$, that is,
\[
\Lambda=\ZZ_{L_x} \times \ZZ_{L_y} \times \ZZ_{L_z}.
\]
We shall say that a site $u=(i,j,k)\in \Lambda$ is even (odd) iff $i+j+k$ is even (odd).
Note that the parity of a site is well-defined for periodic boundary conditions only if all lattice dimensions
 $L_x,L_y,L_z$ are even, since otherwise
a transformation like $i\to i-L_x$ could change the parity of $i+j+k$.
Let $\Lambda_{even}$ and $\Lambda_{odd}$ be the sublattices including all even and odd sites respectively.
Note that each of the sublattices $\Lambda_{even}$ and $\Lambda_{odd}$ can be identified with face-centered cubic (FCC) lattice, see Fig.~\ref{fig:FCC}.
We shall place a qubit at every {\em even} site of the lattice.
Hence the total number of qubits is
\[
n=\frac12\, L_x L_y L_z.
\]
We shall use the notations $X_u\equiv \sigma^x_u$, $Y_u\equiv \sigma^y_u$, $Z_u\equiv \sigma^z_u$ for the single-qubit Pauli operators
acting on a site $u\in \Lambda_{even}$. Let
$\hat{x}=(1,0,0)$, $\hat{y}=(0,1,0)$, and $\hat{z}=(0,0,1)$
be the basis vectors of the lattice.
For every {\em odd} site $u$ we define a stabilizer generator
\be
\label{generator}
S_u =X_{u-\hat{x}}\,  X_{u+\hat{x}}\,  Y_{u-\hat{y}}\,  Y_{u+\hat{y}} \, Z_{u-\hat{z}}\, Z_{u+\hat{z}},
\ee
see Fig.~\ref{fig:generators} and Fig.~\ref{fig:FCC}.
One can easily check that any pair of generators $S_u$, $S_v$ commute,
\[
S_u S_v =S_v S_u \quad \mbox{for all $u,v\in \Lambda_{odd}$}.
\]
Indeed, any pair $S_u$ and $S_v$ may only overlap on a single qubit or on some pair of qubits, see Fig.~\ref{fig:generators}.
In the first case $S_u$ and $S_v$ act on the shared qubit  by the same Pauli operator
and thus commute. In the second case
$S_u$ and $S_v$  anti-commute at each of the two shared qubits, and thus overall they also commute.

\begin{figure}
\centerline{
\includegraphics[height=3cm]{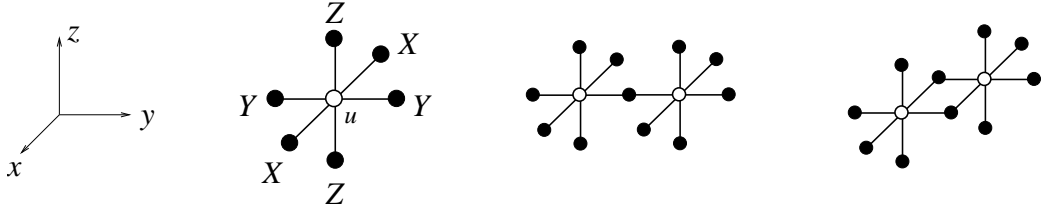}}
\caption{The generator $S_u$ and two possible ways for a pair of generators  to have overlapping supports.
Black dots indicate qubit locations. Generators $S_u$ are centered at sites indicated by open dots.}
\label{fig:generators}
\end{figure}

\begin{figure}[htb]
\centerline{
\includegraphics[height=5cm]{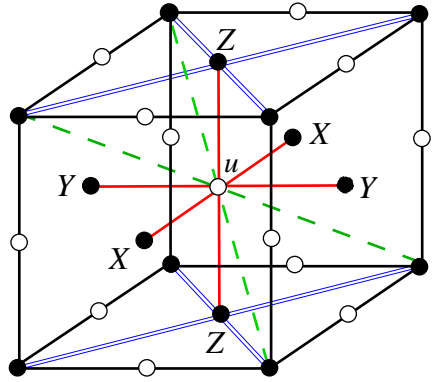}}
\caption{(Color Online) Qubits (black dots) live on the FCC lattice $\Lambda_{even}$. The stabilizer generators (red) $S_u$ are centered on the open dots $u$ in $\Lambda_{odd}$. The double (blue) lines are examples of the six {\em face-diagonals} and the two dashed (green) lines are examples of the four {\em body-diagonals}.}
\label{fig:FCC}
\end{figure}

\begin{figure}[htb]
\centerline{
\includegraphics[height=6cm]{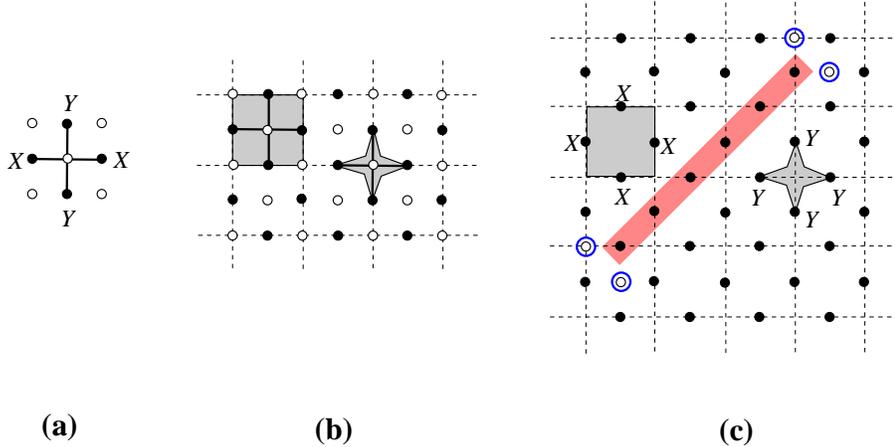}}
\caption{(a) The restriction of a generator $S_u$ onto a  horizontal plane.
(b) The correspondence between even sites lying in the chosen plane
and links of the toric code lattice. (c) The change of basis
$X\to Y$, $Y\to X$, $Z\to -Z$ applied to every horizontal link
maps the in-plane generators $S_u$ to the plaquette ($X$-type) and star ($Y$-type)
operators of the toric code.  A string-like operator
comprised only of Pauli $Z$'s
can commute with both plaquette and star operators only if it follows the diagonal of the lattice
(the pink shaded region).
 Strings that contain $X$ or $Y$ Paulis fail to commute
with the generators $S_u$ located in the adjacent horizontal planes. Hence extending the toric code
to 3D leads to rigidity of strings. Double circles indicate  excitations created near the end-point of the string.}
\label{fig:toric}
\end{figure}

The generators $S_u$ are used to define a local term-wise commuting Hamiltonian~\cite{Chamon05}
\be
\label{H}
H=-\sum_{u\in \Lambda_{odd}}\, S_u.
\ee
This model can be viewed as a natural 3D extension of the 2D toric code Hamiltonian.
We illustrate the correspondence between the two models in Fig.~\ref{fig:toric}.
Our first result concerns the ground state energy and the ground space degeneracy of $H$.
Note that the pairwise commutativity of the generators $S_u$ by itself does not guarantee
that one can minimize all terms in $H$ simultaneously, for example, if some
product of the generators gives $-I$. Before stating the result let us introduce some notations.
For any integers $p,q,r$ we shall use notation $\mathrm{gcd}(p,q,r)$ for the greatest
common divisor of $p,q,r$, that is, the largest integer that divides each of $p,q,r$.
We shall parameterize the lattice dimensions as
\[
L_x=2p_x, \quad L_y=2p_y, \quad L_z=2p_z.
\]
Define the subspace
\[
\calL=\{ |\psi\ra \, : \, S_u\, |\psi\ra =|\psi\ra \quad \mbox{for all $u\in \Lambda_{odd}$}\}.
\]
In Appendix~\ref{sec:appA} we prove the following:
\begin{theorem}
\label{thm:1}
The subspace $\calL$ has dimension $2^{4g}$, where $g={\mathrm{gcd}}(p_x,p_y,p_z)$.
\end{theorem}
Since states from $\calL$ minimize every term in the Hamiltonian, we conclude that $\calL$ is the ground
subspace of $H$.
Theorem \ref{thm:1} shows that the ground space degeneracy of the model Eq.~(\ref{H}) is determined not only by the topology of the underlying manifold (which is always the $3$-torus),  but also by the particular lattice discretization of this manifold. Hence this model is not purely topological in the same sense as the 2D toric code or the 3D models studied in~\cite{HZW:3dmodels}.
It is worth pointing out that the choice of  lattice discretization (even/odd lattice dimensions)
does affect the ground state degeneracy for a  slightly modified 2D toric code model studied by Wen~\cite{Wen03}.
One reason why we call the model Eq.~(\ref{H}) topological is that its
ground state degeneracy  is characterized by exponential insensitivity to weak local perturbations while the spectral gap
above the ground state is stable against such perturbations. We can prove this
by explicitly checking the sufficient conditions for stability of topological phases derived in~\cite{BHM:stab},
see Appendix~\ref{sec:stability}. In particular, we show that the ground subspace $\calL$
is a quantum error correcting code with a distance growing at least linearly with the lattice size.

Our second result is a complete classification of bulk excitations of the model,
see Sections~\ref{sec:strings} and~\ref{sec:SSS}.
Since this classification does not depend on the boundary conditions, it is more natural
to work with the infinite lattice, that is, $\Lambda=\ZZ\times \ZZ\times \ZZ$.
For the infinite lattice the ground state of the model is specified by the eigenvalue equations $S_u=1$ for
all $u\in \Lambda_{odd}$. The ground state is  non-degenerate in the sense
made precise in Appendix~\ref{sec:stability}.
Since all generators $S_u$ pairwise commute and $S_u^2=I$, one can describe any
excited eigenstate of the model by specifying the eigenvalues $S_u=\pm 1$
for each generator. The elementary excitations that we call {\em monopoles}
flip the sign of a single generator.  The following theorem proved in Section~\ref{subs:mem}
shows that
creating an isolated monopole (or any excited spot that contains an odd number of  monopoles) separated from other excitations by distance
$R$ requires an operator acting on roughly $R^2$ qubits.
\begin{theorem}
\label{thm:bound}
Let $M\subseteq \Lambda_{odd}$ be any finite subset containing odd number of sites.
Let $\{P_R\}_{R\ge 1}$ be a sequence of operators
such that $P_R$ creates excitations at every site $u\in M$
and may be some other excitations separated from $M$ by distance at least $R$.
Then $P_R$ must act on $\Omega(R^2)$ qubits in the limit $R\to \infty$.
\end{theorem}
The theorem immediately implies that the parity of the number of monopoles defines a $\ZZ_2$ topological charge,
that is, this quantity cannot be changed by any local operator.  In this respect monopoles are similar to the electric and magnetic excitations in the 2D toric code.
What is more surprising, Theorem~\ref{thm:bound} implies that no string-like operator can create an isolated monopole at its end-point. Indeed, otherwise one would be able
to construct a sequence of string-like operators $\{P_R\}$ as above whose weight grows only linearly with $R$. The absence of string operators associated with monopoles also implies that no local operator can create an isolated  {\em pair} of monopoles from the ground state even if the two monopoles
are separated by a constant distance\footnote{Note however, that the number
of monopoles modulo four does not define a topological charge; for example, a single-qubit
Pauli operator can transform $1$ monopole into $3$ monopoles~\cite{Chamon05}.}.
Indeed, otherwise one would be able to create a  long-range pair of monopoles separated by distance $R$
by combining $O(R)$ short-range pairs, which would require an operator acting only on $O(R)$ qubits.
It also shows that no local operator can induce hopping of monopoles between adjacent sites
unless there are other monopoles nearby.
In Section~\ref{subs:mem} we show that isolated monopoles can be created on the corners of rectangular-shaped membranes, see Fig.~\ref{fig:membrane}. It demonstrates that the lower bound in Theorem~\ref{thm:bound} is tight.
The presence of isolated monopoles which can only created by membrane-like operators is an essential and new feature of the model. As was discussed in~\cite{Chamon05}, it might have interesting physical implications, for example,
the diffusive dynamics of the dilute gas of monopoles might be characterized by an
exponentially small (as a function of the inverse temperature) diffusion constant.

Let us now consider an excited spot $M\subseteq \Lambda_{odd}$ that contains an even number of monopoles. In Section~\ref{subs:SSS} we show how to decompose $M$ into a combination of composite excitations
that we call {\em dipoles} and {\em quadrupoles}.
These excitations are composed of two and four monopoles respectively.
We show that both dipoles and quadrupoles can be created by string-like operators;
in this respect they are similar to electric charges in the 3D surface code~\cite{HZW:3dmodels}.
The unusual feature of the model is that dipoles
can only be created at end-points of {\em rigid strings}, that is, these strings
must be aligned with face-diagonals of the lattice. An example of a rigid string
lying in a horizontal plane and aligned with the face-diagonal $\hat{x}+\hat{y}$
is shown in Fig.~\ref{fig:toric}(c), see also Fig.~\ref{fig:Hstring}. A rigid string creates a dipole near each of its end-points.  The quadrupole-type strings must lie in a pair of adjacent planes orthogonal
to some body-diagonal of the lattice. We shall refer to such pair of planes as a {\em bilayer}.
As we explain in Section~\ref{subs:Texample}, one can identify qubits of any bilayer with a 2D hexagonal
lattice while the generators  $S_u$ centered inside the  bilayer
can be identified with the hexagonal plaquette operators of Kitaev's honeycomb lattice model~\cite{kitaev:anyon_pert}.
Taking a product of elementary link operators  of Kitaev's model
over an arbitrary path on the hexagonal lattice one obtains a string operator creating a
quadrupole near each end-point of the chosen path. We call such strings
{\em flexible bilayer strings}, or simply flexible strings,
since they  can be arbitrarily deformed as long as a string does not leave the bilayer it belongs to.
In Section \ref{subs:SSS} we introduce a complete set of topological charges which we use to
show that  monopoles, dipoles, and quadrupoles are topologically distinct from each other.

Properties of the quantum error correcting code corresponding to the model Eq.~(\ref{H})
are described in Section~\ref{sec:encoding}.
Recall that the main parameters of a quantum code
are the number of physical qubits $n$, the number of logical qubits $k$, and the code distance $d$
which is the number of single-qubit errors needed to destroy the encoded information,
see Section~\ref{sec:encoding} for formal definitions.
Theorem~\ref{thm:1} implies that
the degenerate ground space $\calL$ can be used to encode
$k=4g$ logical qubits into $n=4p_xp_yp_z$ physical qubits.
The best lower bound on the distance $d$ that we can rigorously prove
is $d=\Omega(L)$, where $L$ is the smallest of the lattice dimensions, see
Lemma~\ref{lemma:stability} in Appendix~\ref{sec:stability}.
This lower bound however completely ignores some subtle features of the model
such as the relationships between prime factors of the lattice
dimensions. These features affect the ground state properties in a dramatic way
as can be seen from Theorem~\ref{thm:1}.
Hence we believe that a more favorable scaling of the distance can be achieved
by fine-tuning of the parameters $p_x,p_y,p_z$, see the discussion in Section~\ref{sec:encoding}.

We conclude by discussing some open problems in Section~\ref{sec:open}.

\section{Rigid and flexible strings}
\label{sec:strings}

In this section we describe two classes of string-like operators capable of creating remote clusters of excitations located near the two end-points of a string. Since classification of string operators is a bulk property we shall only consider an infinite lattice.

The strings from the first class which we call rigid strings are straight lines aligned with one of the six face-diagonals of the lattice. The corresponding string operator creates a pair of excitations (a dipole) located near each end-point of the string. The strings from the second class which we call flexible strings lie in a pair  of two adjacent planes orthogonal to one of the four
body-diagonals of the lattice. We call such a pair of planes  a bilayer.  Flexible strings can follow an arbitrary trajectory within the chosen bilayer but they cannot leave it. The corresponding string operator creates four excitations (a quadrupole) at each end-point of the string.

We shall see that string operators associated with closed flexible loops can be expressed in terms of the generators $S_u$. Although a single rigid string cannot be closed into a loop, we shall see that one can combine multiple rigid strings into a closed $3$-valent graph (a string-net) such as a tetrahedron. The corresponding string-net operator can also be expressed in terms of the generators $S_u$.

\subsection{Notations}
Let us remind the reader of some standard notation pertaining to the FCC lattice.
The basis vectors of the 3D simple cubic  lattice  $\Lambda=\ZZ\times \ZZ\times \ZZ$ will be denoted as
\[
\hat{x}=\left[ \ba{c} 1 \\  0 \\   0 \\ \ea \right], \quad
 \hat{y}=\left[ \ba{c} 0 \\  1 \\   0 \\ \ea \right], \quad
 \hat{z}=\left[ \ba{c} 0 \\  0 \\   1 \\ \ea \right].
 \]
Let $a,b,c\in \{1,0,-1\equiv \bar{1}\}$ be a triple of integers.
We shall use the following abbreviations.
\begin{center}
\begin{tabular}{|r|l|}
\hline
& \\
$[abc]$  & a vector $a\hat{x} + b\hat{y} + c\hat{z}$ \\
$[abc]$-plane & a plane orthogonal to the vector $[abc]$ \\
& \\
\hline
\end{tabular}
\end{center}
In this notation the six face-diagonals of the lattice, see Fig.~\ref{fig:FCC}, are
\be
\label{face-diag}
[110], \quad [1\bar{1}0], \quad [101], \quad [10\bar{1}], \quad [011], \quad [01\bar{1}],
\ee
and the four body-diagonals of the lattice are
\be
\label{body-diag}
[111], \quad [1\bar{1}\bar{1}], \quad [\bar{1}1\bar{1}], \quad [\bar{1}\bar{1}1].
\ee

\subsection{Rigid strings and dipoles}
\label{subs:Hformal}

Let $m>0$ be any integer. Consider a set of $m+1$ sites
\[
\gamma=\{ (0,0,0),(1,1,0),\ldots, (m,m,0)\}\subset \Lambda_{even}
\]
and an operator
\[
W(\gamma)=\prod_{u\in \gamma} \, Z_u.
\]
One can easily check that $W(\gamma)$ anti-commutes with exactly four generators $S_u$ for which
$u$ lies in the $[001]$-plane with the $z$-coordinate $k=0$ and has exactly one neighbor in $\gamma$, see Fig.~\ref{fig:Hstring}, namely,
$u=(-1,0,0)$, $u=(0,-1,0)$, $u=(m+1,m,0)$, and $u=(m,m+1,0)$.
The set of sites $\gamma$ is an example of a rigid string.

The other rigid strings associated with the $[010]$-planes and $[100]$-planes can be obtained from this example by applying the lattice symmetries and replacing $Z_u$ by $Y_u$ and $X_u$ respectively.
Let $h$ be one of the six face-diagonals, see  Eq.~(\ref{face-diag}) or
Fig.~\ref{fig:FCC}. We can formally define the rigid strings of type $h$ as follows.
\begin{dfn}
A sequence of sites $\gamma=(u_0,u_1,\ldots,u_m)$ with $u_i \in \Lambda_{even}$ is called
an rigid string of type $h$ iff
$u_{i+1}-u_i=h$ for all $i=0,\ldots,m-1$.
The sites $u_0$ and $u_m$ are called the end-points of the string.
\label{dfn:Hstring}
\end{dfn}
We shall refer to a pair of excitations associated with an end-point of an rigid string of type $h$
 as a {\em dipole} of type $h$.


\begin{figure}
\centerline{
\includegraphics[height=5cm]{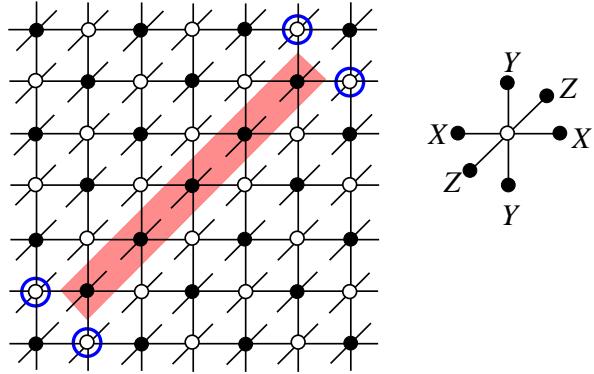}}
\caption{Example of a rigid string $\gamma$ lying in the $[001]$-plane.  The string operator $W(\gamma)$ acts by $Z$ on the qubits in the shaded region. Double circles near the end-points of the string indicate locations of excitations created by the string operator $W(\gamma)$.
A pair of excitations located near each end-point of the string is called a dipole.}
\label{fig:Hstring}
\end{figure}

We may ask whether it is possible to define a closed rigid string  $\gamma$
such that the corresponding string operator $W(\gamma)$
is a  product of the  generators $S_u$ (similar to  closed loops in the 2D toric code).
Due to the rigidity of the strings one cannot close a string in a single plane. However one can form a three-dimensional object, a tetrahedron, by using all six types of rigid strings interconnected with each other
to form a closed string-net, see Fig.~\ref{fig:tetra}. Note that the
action of the three rigid strings at the qubits located at vertices of
the tetrahedron cancels since $XYZ=iI$. Thus the smallest tetrahedron operator in a cube of dimensions $3\times 3 \times 3$ equals the single generator $S_v$ at the center of the cube, depicted in Fig.~\ref{fig:FCC}.

\begin{figure}
\centerline{
\includegraphics[height=4cm]{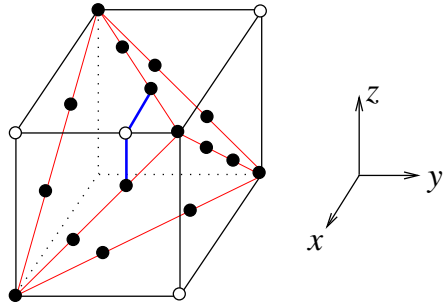}}
\caption{(Color Online) Construction of a closed rigid string-net or tetrahedron operator. The six edges of the tetrahedron $T$ are rigid strings of the six possible types lying on the faces of the cube. The string-net operator $W(T)$ is the product of the six corresponding string operators.
Note that $W(T)$ acts trivially at the qubits located at vertices of $T$ since the triple
of strings incident to any vertex of $T$ cancel each other through the identity $XYZ\sim I$.
The open dot with two incident blue lines indicates the location of excitations created by a pair of string operators.
The two excitations cancel each other.}
\label{fig:tetra}
\end{figure}

Here is how we define a general tetrahedron operator. Let $C\subseteq \Lambda$ be a some cube of the lattice such that all eight vertices of $C$ belong to $\Lambda_{even}$.  Let $T$ be a tetrahedron formed by four vertices of $C$ such that edges of $T$ are diagonals of faces of $C$, see Fig.~\ref{fig:tetra}. By construction, all sites lying on the six
edges of $T$ are even. We can denote the six edges of $T$ as $\gamma^x_1,\gamma^x_2,\gamma^y_1,\gamma^y_2,\gamma^z_1,\gamma^z_2$,
where  $\gamma^\alpha_1$ and $\gamma^\alpha_2$ is the pair
of edges orthogonal to the axis $\alpha$. We define the tetrahedron operator $W(T)$ as
\[
W(T)=\prod_{\alpha =x,y,z} \, \prod_{j=1,2} \, \prod_{u\in \gamma^\alpha_j} \, \sigma^\alpha_u.
\]
Each rigid string operator involved in $W(T)$ creates a pair of excitations near the end-point of the string, i.e., near some vertex of the tetrahedron $T$. However, as one can easily check, the excitations created by a triple of strings incident to any vertex of $T$ cancel each other and thus the tetrahedron operator $W(T)$ commutes with all generators $S_v$.
A simple inspection shows that $W(T)$ can be represented as a product of {\em all} generators $S_v$
in the interior of $T$ (up to a phase factor):
\be
\label{O(T)}
W(T)\sim \prod_{v\in T\cap \Lambda_{odd}} \, S_v.
\ee
The tetrahedron operators will be important in Section \ref{subs:mem}, Lemma \ref{lemma:monopole}, where we use them to detect an isolated monopole.


\subsection{Flexible strings and quadrupoles}
\label{subs:Texample}

\begin{figure}
\centerline{
\includegraphics[height=5cm]{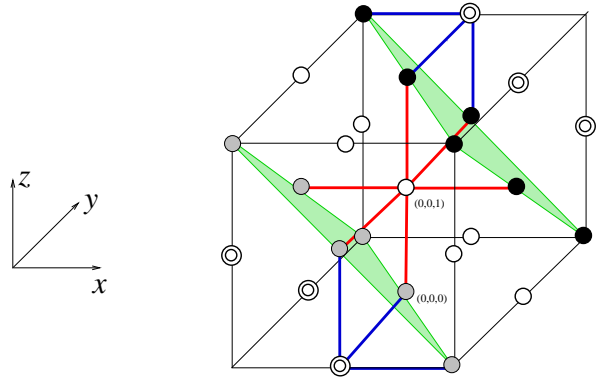}}
\caption{(Color Online) Open dots indicate centers of generators $S_u$ located in the plane $\Sigma_1$.
Gray and black dots indicate qubits located in the planes $\Sigma_0$ and $\Sigma_2$ respectively.
The union of the planes $\Sigma_0\cup \Sigma_2$ is called a bilayer (shaded green area).
A generator  (red lines) centered in the plane $\Sigma_1$ can be viewed as six-body hexagonal plaquette
 operators acting on this bilayer. The double circles indicate centers of generators $S_u$
 located in the planes  $\Sigma_{-1}$ and $\Sigma_3$. These generators
  touch three qubits in $\Sigma_0$ or three qubits in  $\Sigma_2$ and can be represented as star operators, see Figs.~\ref{fig:hexagons} and \ref{fig:stars}.}
\label{fig:hexplane}
\end{figure}

The construction of flexible strings is more involved. Let us start by constructing such a string lying in a pair of $[111]$-planes $\Sigma_0$, $\Sigma_2$,
where
\be
\Sigma_\alpha=\{ (i,j,k)\in \Lambda \, : \, i+j+k=\alpha\}.
\label{eq:siga}
\ee
We identify the union $\Sigma_0\cup \Sigma_2$ with the 2D hexagonal lattice
$\hexlat$ by projecting the sublattices $\Sigma_0$, $\Sigma_2$ onto some fixed $[111]$-plane,
see e.~g.~Fig.~\ref{fig:hexplane}.
We shall refer to the pair of planes $\Sigma_0\cup \Sigma_2$ as a {\em $[111]$-bilayer}, or simply
a bilayer.

It will be convenient to color the sites of $\hexlat$ in black and gray using the
checkerboard coloring as shown in Fig.~\ref{fig:hexagons}. Then we have the following correspondence between the union of the sublattices $\Sigma_0\cup \Sigma_1\cup \Sigma_2$ and the 2D lattice $\hexlat$:
\begin{center}
\begin{tabular}{c|c}
3D  & 2D  \\
\hline
$\Sigma_0$ & gray sites \\
\hline
$\Sigma_1$ & centers of hexagons (white sites)\\
\hline
$\Sigma_2$ & black sites \\
\hline
\end{tabular}
\end{center}
\begin{figure}
\centerline{
\includegraphics[height=4cm]{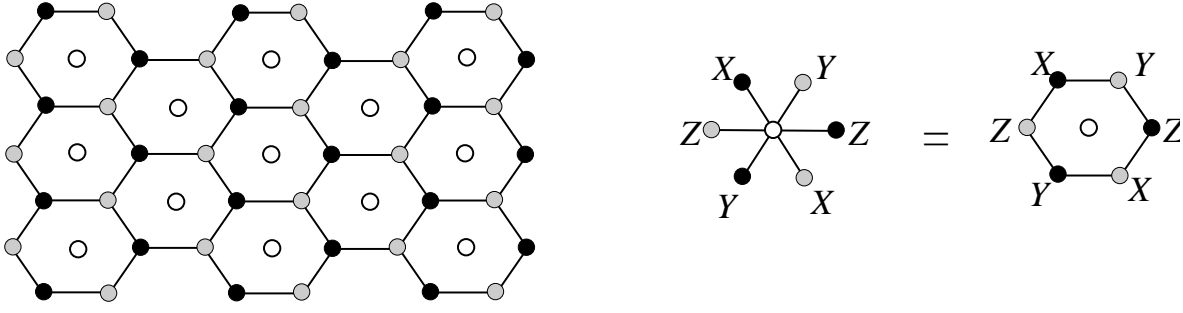}}
\caption{The correspondence between the union of the sublattices
$\Sigma_0\cup \Sigma_1 \cup \Sigma_2$ and the 2D hexagonal lattice.
The sites  indicated by open dots (centers of hexagons) represent $\Sigma_1$.
The  sites indicated by gray and black dots represent $\Sigma_0$ and $\Sigma_2$ respectively.
The generators $S_u$ located at $\Sigma_1$ can be viewed as plaquette
operators associated with the hexagons as shown on the right. }
\label{fig:hexagons}
\end{figure}

Define a triple of vectors
\be
\label{lambdas}
\lambda^x=[011], \quad \lambda^y=[101],
\quad \lambda^z=[110].
\ee
Consider any pair of sites $u\in \Sigma_0$ and $v\in \Sigma_2$. By construction, $(u,v)$ is a link
of the hexagonal lattice iff $v=u+\lambda^\alpha$ for some $\alpha\in \{x,y,z\}$.
We shall say in this case that $(u,v)$ is an $\alpha$-link.

Define the {\em link operator} (as in Kitaev's honeycomb model \cite{kitaev:anyon_pert})
\[
K_{(u,v)}=\sigma^\alpha_u \sigma^\alpha_v \quad \mbox{if $(u,v)$ is an $\alpha$-link}.
\]
Given any pair of links $e\ne e'$ the operators $K_e$ and $K_{e'}$
commute iff $e,e'$ do not overlap and anti-commute
iff $e$ and $e'$ share exactly one vertex.
For any plaquette (hexagon) $p$ and and any site $w$ of $\hexlat$ define plaquette and star operators
\be
\label{BAoperators}
B_p=\prod_{e\in \partial p}\; K_e, \quad A_w=i\prod_{e\in {\mathrm{star}}(w)}\; K_e,
\ee
see Fig.~\ref{fig:stars}.
\begin{figure}
\centerline{
\includegraphics[height=3cm]{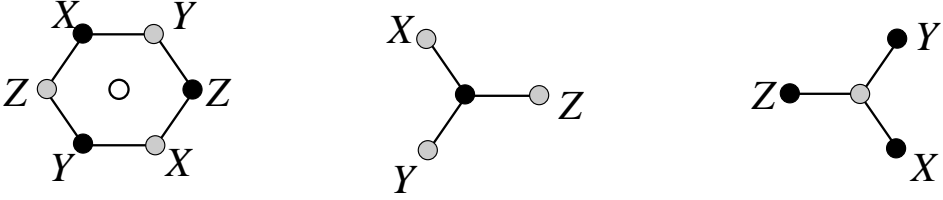}}
\caption{Plaquette and star operators.}
\label{fig:stars}
\end{figure}
Here $\partial p$ is the set of six links forming the boundary of $p$, while ${\mathrm{star}}(w)$ is
the set of three links incident to $w$. Note that the action of $A_w$ onto the site $w$ is proportional to $iX_w Y_w Z_w \sim I$, that is, $A_w$ acts only on the nearest neighbors of $w$.
In particular, $A_w$ acts only on gray (black) sites iff $w$ is black (gray),
see Fig.~\ref{fig:stars}.
The only generators that can act non-trivially on qubits of $\hexlat=\Sigma_0\cup \Sigma_2$
are those located in the planes
$\Sigma_{-1}$, $\Sigma_1$, and $\Sigma_3$.  One can easily check that
the action of these generators on the qubits of the hexagonal lattice can be described as follows:
\begin{center}
\begin{tabular}{c|c}
3D  & 2D  \\
\hline
$\Sigma_{-1}$ & star operators centered at black sites  \\
\hline
$\Sigma_1$ & plaquette operators \\
\hline
$\Sigma_3$ & star operators centered at gray sites \\
\hline
\end{tabular}
\end{center}
We have the commutation rules
\be
\label{link-plaq}
[K_e,B_p]=0 \quad \mbox{for all links $e$, for all plaquettes $p$}.
\ee
Furthermore,
\be
[K_e,A_w]=0 \quad \mbox{iff $e$ is incident to $w$ or $e$ does not overlap with ${\mathrm{star}}(w)$}.
\ee
Alternatively,
\be
\label{link-star}
K_eA_w=-A_w K_e \quad \parbox[t]{10cm}{iff $e$ is incident to a neighbor of $w$ but not incident to $w$}.
\ee
Consider any pair of sites $u,v\in \hexlat$ (black or gray) and let $\gamma$ be any path on $\hexlat$
connecting $u$ and $v$. Define a string operator
\[
W(\gamma)=\prod_{e\in \gamma}\, K_e.
\]
\begin{figure}
\centerline{
\includegraphics[height=4cm]{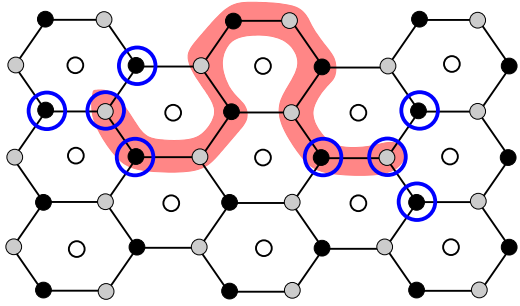}}
\caption{Example of a flexible bilayer string. Double circles near each end-point of the string indicate the locations of excitations created by the string operator $W(\gamma)$. Recall that excitations located at gray and black sites correspond to generators $S_u$ with $u\in \Sigma_3$ and $u\in \Sigma_{-1}$. The group of four excitations
located near each end-point of the string is called a quadrupole. }
\label{fig:Tstring}
\end{figure}
It follows from Eq.~(\ref{link-plaq}) that $W(\gamma)$ commutes with all plaquette operators $B_p$.
In addition, since $\gamma$ has even number of links incident to any site $w\neq u,v$,
we infer from Eq.~(\ref{link-star}) that $W(\gamma)$ commutes with star operators
$A_w$ unless $w$ is located near the end-points $u$ or $v$. We shall be mostly interested in the
case when the end-points $u,v$ are sufficiently far apart such that the commutation rules between
$W(\gamma)$ and star operators can be analyzed independently near each end-point of $\gamma$.
Then one can easily check that $W(\gamma) A_w=-A_w W(\gamma)$ iff
$w$ is the end-point of $\gamma$, or $w$ is a neighbor of the end-point of $\gamma$.

It is clear that similar strings can be constructed in the bilayers orthogonal to any body-diagonal $t$,
see Eq.~(\ref{body-diag}).
Such a flexible bilayer string lies in a pair of even planes orthogonal to some body-diagonal $t$ separated by a single odd plane as Fig.~\ref{fig:hexplane} demonstrates. The group of four excitations located near each end-point of the string will be referred to as a {\em quadrupole of type $t$}. More formally, we define these string operators as


\begin{dfn}[\bf Flexible bilayer strings]
Let $t=(t_x,t_y,t_z)$ be any body-diagonal.
A sequence of sites $\gamma=(u_0,u_1,\ldots,u_m)$ with $u_i \in \Lambda_{even}$ is called
a flexible $[t]$-bilayer string iff there exists a sequence $\alpha_0,\ldots,\alpha_{m-1}\in \{x,y,z\}$ and $\epsilon=\pm 1$ such that
$u_{j+1}=u_j + (-1)^j \epsilon  \, \lambda(\alpha_j)$ for all $j=0,\ldots,m-1$, where
the vectors $\lambda(\alpha)$ are defined as
\[
\lambda(x)=\left[ \ba{c} 0 \\ t_y \\ t_z \\ \ea \right],
\quad
\lambda(y)=\left[ \ba{c} t_x \\ 0 \\ t_z \\ \ea \right],
\quad
\lambda(z)=\left[ \ba{c} t_x \\ t_y \\ 0 \\ \ea \right].
\]
The sites $u_0$ and $u_m$ are called the end-points of the string. We can define a string operator $W(\gamma)$ associated with $\gamma$ as
\be
W(\gamma)=\prod_{j=0}^{m-1} \sigma_{u_j}^{\alpha_j} \sigma_{u_{j+1}}^{\alpha_j}.
\ee
\end{dfn}

We assume that the separation between the end-points of $\gamma$ is sufficiently large so one can derive the commutation rules near each end-point independently. The commutation rules between $W(\gamma)$ and the generators $S_v$ are then summarized by the following Lemma~\ref{lemma:Tstring}.

\begin{lemma}
\label{lemma:Tstring}
Let $\gamma=(u_0,u_1,\ldots,u_m)$ be a flexible $[t]$-bilayer string.
Then $W(\gamma) S_v = -S_v W(\gamma)$ iff $v$ is one of the eight sites listed below:
\[
u_0-\epsilon\left[ \ba{c} t_x \\ 0 \\  0 \\ \ea \right],
\quad
u_0-\epsilon\left[ \ba{c} 0 \\ t_y \\  0 \\ \ea \right],
\quad
u_0-\epsilon\left[ \ba{c} 0 \\ 0 \\  t_z \\ \ea \right],
\quad
u_0+\epsilon\left[ \ba{c} t_x \\ t_y \\  t_z \\ \ea \right],
\]
and
\[
u_m+\epsilon (-1)^m \left[ \ba{c} t_x \\ 0 \\  0 \\ \ea \right],
\quad
u_m + \epsilon (-1)^m \left[ \ba{c} 0 \\ t_y \\  0 \\ \ea \right],
\quad
u_m + \epsilon (-1)^m \left[ \ba{c} 0 \\ 0 \\  t_z \\ \ea \right],
\quad
u_m - \epsilon (-1)^m \left[ \ba{c} t_x \\ t_y \\  t_z \\ \ea \right].
\]
\end{lemma}

Again, we can ask how one can construct closed flexible strings. It is not hard to see that such closed strings which lie, say, in a $[111]$-bilayer, can be constructed from the hexagonal plaquette operators $B_p$. Specifically, let $p$ be any plaquette of $\Omega_{hex}$ and $v\in \Sigma_1$ be the center of $p$.
By definition of the plaquette operators we have $B_p=S_v$ and $B_p$ can be regarded as a string operator, $B_p=W(\gamma)$ where $\gamma=\partial p$ is the closed flexible string
that consists of the six links lying on the boundary of $p$. Similarly, if $\gamma$ is a closed flexible string without self-intersections, and $\mathrm{Int}(\gamma)$ is the set of plaquettes encircled by $\gamma$, we get the identity
\be
\label{closedTstring}
W(\gamma) \sim \prod_{p\in \mathrm{Int}(\gamma)}\, B_p.
\ee
We will use these closed flexible  strings in Section~\ref{subs:SSS} to characterize the topological charges.




\section{Classification of excitations}
\label{sec:SSS}

The purpose of this section is to classify the excitations of our 3D model. We will first consider how single isolated excitations which we call monopoles can emerge. We say that an excited state $|\psi\ra$ has an isolated  monopole at a site $u\in \Lambda_{odd}$
iff $S_u\, |\psi\ra =-|\psi\ra$ and $S_v\, |\psi\ra =|\psi\ra$ for all sites $v\ne u$ in a sufficiently large neighborhood of $u$.
Note that the string operators described in Section~\ref{sec:strings} can only create an {\em even} number of excitations
(a dipole or a quadrupole) near each end-point of a string.
Therefore a natural question is whether some more complicated  string-like operators  could create a single monopole.
In Section~\ref{subs:mem} we answer this question in the negative: we prove in Lemma~\ref{lemma:monopole} that the minimum weight of a Pauli operator creating an isolated monopole grows as $\Omega(R^2)$, where $R$ is the separation between the monopole and the nearest excitation. We then use the lemma to prove Theorem~\ref{thm:bound}.

The goal of Section \ref{subs:SSS} is to identify all possible excitations of the model that are
topologically non-trivial, that is, excitations that cannot be created locally from the ground state.
We describe a complete set of $\ZZ_2$ topological charges such that an excitation is
topologically non-trivial iff it is characterized by non-zero value of some topological charge.

In contrast to 2D systems that can always be described by a finite number of topological charges, 3D systems may possess an infinite number of topological charges. The simplest example of such behavior is a 3D stack of non-interacting 2D toric codes, in which each copy of the toric code contributes a constant number of charges.  We shall see that in our model the number of topological charges is also infinite, but it becomes finite after factoring out quadrupole excitations.

A simple corollary of our classification is the proof that rigid strings and flexible  bilayer strings are the only possible strings in the theory. More precisely, let $P$ be a Pauli operator creating an excited spot $M\subseteq \Lambda_{odd}$ with $m=O(1)$ excitations separated from all other excitations by a distance $R\gg 1$ (we are interested in the limit $R\to\infty$). If $m$ is odd,
Theorem~\ref{thm:bound}  shows that $|P|=\Omega(R^2)$, that is, $P$ cannot be reduced to a constant number of string-like operators for any reasonable definition of a string. On the other hand, if $m$ is even, we will prove that $M$
can be decomposed into a constant number of dipoles and quadrupoles, see Section~\ref{subs:decomposition},
and thus $M$ can be
created using a constant number of rigid strings and flexible  strings.

\subsection{Monopoles and membrane operators}
\label{subs:mem}

Isolated monopoles can be created by operators supported on rectangular-shaped membranes
such that there is one monopole sitting near each of the corners of the membrane, see Fig.~\ref{fig:membrane}.
\begin{figure}
\centerline{
\includegraphics[height=4cm]{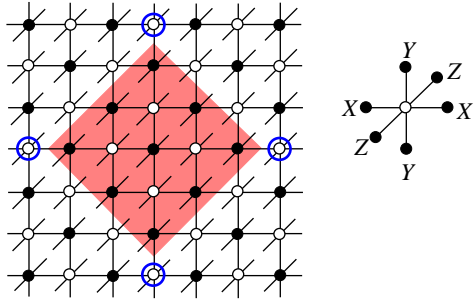}}
\caption{Double circles indicate locations of four isolated monopoles created by a membrane
operator. The membrane $M$ consists of all even (black) sites that belong to the shaded region.
The corresponding membrane operator is  $W(M)=\prod_{u\in M} \, Z_u$. }
\label{fig:membrane}
\end{figure}
For example, consider a set of sites
\be
\label{mem1}
M=\{ (i,j,0)\in \Lambda_{even}\, : \, |i| + |j|\le R\}
\ee
and let
\be
\label{mem2}
W(M)=\prod_{u\in M} Z_u.
\ee
Assume for simplicity that $R$ is even. Then one can easily check that $W(M)$
creates four isolated monopoles at the sites listed below.
\be
W(M) S_v=-S_v W(M) \quad \mbox{iff} \quad
v\in \{ (R+1,0,0), \quad (0,R+1,0), \quad (-R-1,0,0), \quad (0,-R-1,0)\}.
\ee
The operator $W(M)$ is an example of a membrane operator of weight $O(R^2)$.
We shall now prove that the weight of any operator creating an isolated monopole must grow quadratically with $R$.
To illustrate the idea of the proof, let us first consider a 2D analogue of this problem.
 Suppose a Pauli operator $P$ creates an electric charge $e$  in the 2D toric code
separated from other excitations by a distance $R$. The charge  $e$ can be detected
by moving a magnetic charge $m$ over an arbitrary closed loop $\gamma$
encircling $e$. The corresponding string operator $W(\gamma)$ must anti-commute with $P$
and hence the supports of $P$ and $W(\gamma)$ must share at least one qubit.
Choosing a  family of $m=\Omega(R)$ pairwise disjoint loops $\gamma_1,\ldots,\gamma_m$ of increasing radius
that encircle $e$,
we conclude that the support of $P$ shares at least one qubit with each loop $\gamma_i$,
and thus $|P|\ge m=\Omega(R)$. In our 3D model we shall detect isolated monopoles
using a family of $\Omega(R^2)$ closed rigid strings or tetrahedron operators, see Section~\ref{subs:Hformal}.
Since a tetrahedron has a one-dimensional support, we can choose $\Omega(R^2)$ pairwise
disjoint tetrahedra of linear size $R$ enclosing the monopole to be detected (in the proof of the Lemma the tetrahedra are in fact not disjoint but overlap on $O(1)$ qubits). It yields a lower
bound $\Omega(R^2)$ on the weight of the operator $P$ creating the monopole.
Note that using operators with membrane-like support to detect the monopole would not be good enough, since we could only construct only $O(R)$ non-overlapping membranes which enclose the monopole. Similarly,
using closed flexible loops to detect the monopole would not be good enough, since
one can only construct $O(R)$ encircling, non-overlapping, flexible loops in the bilayer that contains the monopole.

\begin{lemma}
\label{lemma:monopole}
Let $P$ be a Pauli operator such that $PS_u=-S_u P$ for some site $u\in \Lambda_{odd}$.
Suppose $PS_v=S_vP$ for all sites $v\ne u$ within distance $R$ from $u$.
Then $|P|\ge cR^2$ for some constant $c$.
\end{lemma}
\begin{proof}
Suppose one can construct $m=\Omega(R^2)$ operators $O_1,\ldots,O_m$ with the following properties:
\begin{itemize}
\item Each operator $O_j$ anti-commutes with $P$
\item Each qubit is acted on by at most $c=O(1)$ operators $O_j$
\end{itemize}
Note that each qubit in the support of $P$ can be responsible for anti-commutativity
with at most $c$ operators $O_j$. Since $P$ anti-commutes with all $m$ operators $O_j$,
we conclude that $c|P|\ge m$, that is, $|P|=\Omega(R^2)$. Let us explain how to choose the desired operators $O_j$. Let $B$ be a ball of radius $R$ centered at the site $u$ occupied by the monopole.
For any tetrahedron operator $W(T)$ constructed as in Section \ref{subs:Hformal} such that $u\in T\subseteq B$, one has the desired anti-commutation rule:
\be
\label{mdetect}
P W(T) = -W(T) P \quad \mbox{for all $T\subseteq B$ such that $u\in T$}.
\ee
Here we used the fact that $W(T)$ can be represented as the product of all generators
$S_u$ in the interior of $T$.
It remains to choose sufficiently many tetrahedra with the right properties.
Let $\calT$ be a set of tetrahedra $T$ of size $R/10$
such that $u\in T$ (and thus $T\subseteq B$).
Obviously, $|\calT|=\Omega(R^3)$.
Let us say that a pair of tetrahedra $T,T'\in \calT$ is {\em bad}
iff some edge of $T$ and some edge of $T'$ lie on the same line (and hence they can possibly overlap on $\Omega(R)$ qubits). Otherwise, we shall say that $T,T'$ is a good pair. Note that the supports of a good pair of tetrahedra $T$ and $T'$ overlap on the intersection of non-parallel edges, hence on $O(1)$ qubits. For any fixed $T\in \calT$ there are at most $O(R)$
tetrahedra $T'\in \calT$ such that the pair $T,T'$ is bad.
It follows that there exists a subset of $m=\Omega(R^2)$ tetrahedra
$T_1,\ldots,T_m\in \calT$ such that any pair $T_i,T_j$ is good: we can construct such a subset by
picking a first tetrahedron $T_1$, eliminating all tetrahedra which are bad for $T_1$ and picking
the next tetrahedron from the remaining set etc. At every step one looses $O(R)$ tetrahedra, so starting with $\Omega(R^3)$ tetrahedra allows us to pick $m=\Omega(R^2)$ pairwise good ones.
Choosing $O_j=W(T_j)$, $j=1,\ldots,m$ concludes the proof.
 \end{proof}

Theorem~\ref{thm:bound} is a straightforward generalization of the above lemma.
Indeed,   suppose first that $P_R$ is a Pauli operator.
If $R$ is much larger than the size of the excited spot $M$, one can use tetrahedra of
size $R$ to construct operators $W(T)$ satisfying Eq.~(\ref{mdetect}) and thus prove
the bound $|P_R|=\Omega(R^2)$.
Suppose now that $P$ is an  arbitrary (non-Pauli) operator.
Let $|\psi\ra$ be the ground state and $|\phi\ra=P\, |\psi\ra$ be an excited state
with an isolated monopole, that is, $S_u\, |\phi\ra=-|\phi\ra$ and $S_v\, |\phi\ra=|\phi\ra$
for all sites $v\ne u$ within distance $R$ from $u$.
Decomposing $|\phi\ra$ in the common eigenbasis of the generators $S_v$, $v\in \Lambda_{odd}$,
one can find an eigenstate $|\psi'\ra$ such that $\la \psi'|P|\psi\ra \ne 0$
and $|\psi'\ra$ contains an isolated monopole at $u$. Expanding $P$ in the Pauli
basis we conclude that at least one Pauli operator $P'$ in this expansion must satisfy
$\la \psi'|P'|\psi\ra \ne 0$. Then by Lemma~\ref{lemma:monopole} one must have
$|P'|=\Omega(R^2)$ and thus $P$ acts on $\Omega(R^2)$ qubits.

\subsection{A complete classification}
\label{subs:SSS}

An excited eigenstate of the Hamiltonian Eq.~(\ref{H}) can be described
by a {\em syndrome} $s$ that assigns an eigenvalue $(-1)^{s(u)}$ to any generator $S_u$.
We shall only consider excited states with finite energy (recall that in this section we consider
an infinite lattice). More formally, let $\calP$ be the group of Pauli operators (with possible infinite support).
A function $s\colon \Lambda_{odd} \rightarrow \{0,1\}$ is called a syndrome iff
$s(u)=1$ only for finitely many $u$ and
there exists a Pauli operator $P \in \calP$ such that
\be
\label{syndrome}
S_u P =(-1)^{s(u)} P S_u \quad \mbox{for all $u\in \Lambda_{odd}$}.
\ee
We say that $P$ causes the syndrome $s$.
Note that the commutation between $S_u$ and $P$ is well-defined
even if $P$ has infinite support (although the commutation between Pauli operators
with infinite support is not  well-defined).
One can also regard syndromes  as binary strings
with a finite Hamming weight whose bits are labeled by sites $u\in \Lambda_{odd}$.
A bitwise sum of strings $s,s'$ will be denoted $s\oplus s'$.
We shall group syndromes into equivalence classes called superselection sectors
using the following definition.
\begin{dfn}
\label{def:SSS}
Two syndromes $s$ and $s'$ are equivalent iff $s\oplus s'$ is caused by a Pauli operator $P \in \calP$ with finite support. An equivalence class of syndromes is called a superselection sector. A syndrome is called topologically trivial iff it is equivalent to the trivial syndrome $s(u)=0$ for all $u$.
\end{dfn}

In Section \ref{subs:Texample} we have constructed closed flexible loops
lying in $[111]$-bilayers.  We shall now need to construct such closed loops for all other bilayers.
Let  $t=(t_x,t_y,t_z)$ be any body-diagonal, see Eq.~(\ref{body-diag}),
and $\alpha$ be any  odd integer. Define a plane
\[
\Sigma_{t,\alpha} =\{ (i,j,k)\in \Lambda_{odd} \, : \,
it_x + jt_y + kt_z = \alpha\}.
\]
For any given syndrome $s$ we define a family of  topological $\ZZ_2$ charges as
\be
\label{theta}
\theta_{t,\alpha}(s)=\sum_{u\in \Sigma_{t,\alpha}} \, s(u)\; \modtwo.
\ee
Note that the sum above is well-defined since, by definition, $s$ has finite support.


We will now show that these topological charges $\theta_{t,\alpha}$ form a complete set,
that is,  two syndromes $s$, $s'$ are equivalent, see Definition \ref{def:SSS}, iff $\theta_{t,\alpha}(s)=\theta_{t,\alpha}(s')$
 for all $t$ and $\alpha$. By linearity, it suffices to prove the following.
  \begin{lemma}
\label{lemma:complete}
A syndrome $s$ can be  caused by a Pauli operator $P$ with finite support if and only if for all odd $\alpha$ and all body-diagonals $t$, one has $\theta_{\alpha,t}(s)=0$.
\end{lemma}
To prove the ``only if'' direction we note that  any Pauli operator $P$ with finite support can be enclosed by a sufficiently large closed flexible loop lying in any $[t]$-bilayer. The corresponding string operator
$W_{t,\alpha}(\gamma)$ must  commute with $P$ since their supports are disjoint.
 On the other hand, $W_{t,\alpha}(\gamma)$ coincides with the product of generators $S_v$ over all
 plaquettes of the bilayer that are encircled by $\gamma$, see Section~\ref{subs:Texample}.
 Hence the commutation between $P$ and $W_{t,\alpha}(\gamma)$ is controlled by the topological
 charge $\theta_{t,\alpha}(s)$, and we conclude that $\theta_{t,\alpha}(s)=0$.
 To prove the other direction, we manipulate the syndrome by operators with finite support to effectively annihilate it.

\begin{proof}
Suppose $P\in \calP$ has finite support and let $s$ be its syndrome.
Let us show that $\theta_{t,\alpha}(s)=0$.
Since $P$ has finite support, we can take any sufficiently large flexible loop $\gamma$
 lying in the $[t]$-bilayer $\Sigma_{t, \alpha-1}\cup \Sigma_{t, \alpha+1}$ such that $\gamma$
 encloses the support of $P$. Obviously, the corresponding closed-string operator $W_{t,\alpha}(\gamma)$ commutes with $P$, since their supports are disjoint. On the other hand, the  operator $W_{t,\alpha}(\gamma)$
 is a product of the stabilizer generators in its
 interior ${\rm Int}(\gamma)$, Eq.~(\ref{closedTstring}), and thus $P W_{t,\alpha}(\gamma)=(-1)^{\theta_{t,\alpha}(s)}W_{t,\alpha}(\gamma)P$, i.~e.~the closed-string operator picks up the topological charge of $P$. Thus the topological charge must be zero, $\theta_{t,\alpha}(s)=0$.

Now we prove the other direction, namely if for a syndrome $s$ all charges $\theta_{t,\alpha}(s)=0$, then
$s$ can be caused by  a Pauli operator  $P$ with a finite support.
Applying single-qubit $X$-errors we can shift the support of $s$
onto the pair of $[010]$-planes with $y$-coordinates $j=0$ and $j=1$. Applying single-qubit $Y$-errors we can shift the
support of $s$ in each of these planes onto a pair of adjacent $[001]$-lines $i=0$ and $i=1$.
Since the composition of all $X$- and $Y$-errors used above has finite support,
it does not change functions $\theta_{t,\alpha}$ (see the first part of the proof).
 Hence we can assume that $s$ has support only on the union  of  lines
 $L_{00}$, $L_{01}$, $L_{10}$, and $L_{11}$, where
\[
L_{ab}=\{ (i,j,k)\in \Lambda_{odd} \,: \, i=a, j=b\}.
\]
Suppose that $s\ne 0$.
Let $k$ be the largest integer $k$ such that $s(a,b,k)=1$ for some $a,b\in \{0,1\}$.
Let $u=(a,b,k)$.
Consider four cases.

\noindent
{\em Case~1:} $(a,b)=(0,0)$. Then the plane $\Sigma_{t,k}$
with $t=(-1,-1,1)$ contains a single non-zero syndrome bit $s(u)=1$, $u=(0,0,k)$.
It means that $\theta_{t,k}(s)=1$ which is a contradiction.

\noindent
{\em Case~2:} $(a,b)=(0,1)$. Then the plane $\Sigma_{t,-k-1}$
with $t=(1,-1,-1)$  contains a single non-zero syndrome bit $s(u)=1$, $u=(0,1,k)$.
It means that $\theta_{t,-k-1}(s)=1$ which is  a contradiction.

\noindent
{\em Case~3:} $(a,b)=(1,0)$. Then the plane
$\Sigma_{t,-k-1}$ with $t=(-1,1,-1)$ contains a single non-zero syndrome but $s(u)=1$,
$u=(1,0,k)$. It means that $\theta_{t,-k-1}(s)=1$ which is a contradiction.

\noindent
{\em Case~4:} $(a,b)=(1,1)$. Then the plane $\Sigma_{t,k+2}$
with $t=(1,1,1)$  contains a single non-zero
syndrome bit $s(u)=1$, $u=(1,1,k)$.
It means that $\theta_{t,k+2}(s)=1$  which is a contradiction.

Summarizing, we get a contradiction unless $s=0$. It means that the original
syndrome can be caused by a Pauli operator with a finite support.
\end{proof}

\subsection{Decomposition into dipoles and quadrupoles}
\label{subs:decomposition}


Let $s$ be any syndrome (an excited state)
such that the total number of excitations in $s$ is even.
In this section we show how to decompose $s$ into a combination of dipoles
and quadrupoles. To illustrate the usefulness of our formalism we show how to
create a `dislocation' in a rigid string that can be composed only of quadrupole charges located in the neighborhood of the
dislocation.

Let us first determine the configuration of topological charges, Eq.~(\ref{theta}),
corresponding to a single quadrupole, a dipole and a monopole, the results are summarized in Fig.~\ref{fig:topcharges}.

\begin{figure}[htb]
\centerline{
\includegraphics[height=4cm]{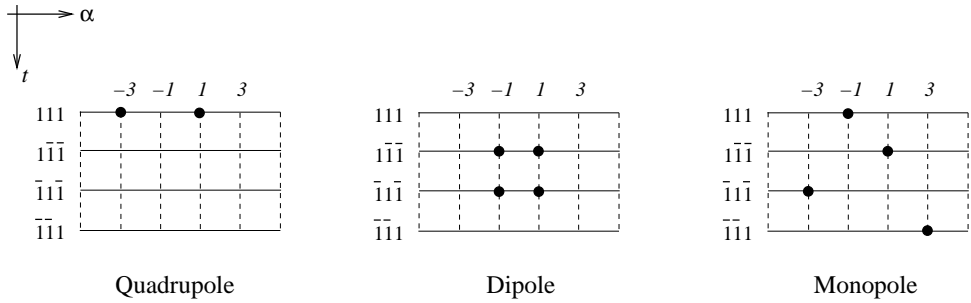}}
\caption{Given a syndrome $s$, the corresponding superselection sector
can be described by a diagram where the black dots indicate non-zero topological $\ZZ_2$ charges $\theta_{t,\alpha}$.
Here $t$ labels body-diagonals and $\alpha$ is an arbitrary odd integer.
Examples of diagrams representing a single quadrupole of type $[111]$, a single dipole of type $[110]$, and a single monopole
located at site $u=(0,-2,1)$
are shown.  For any syndrome $s$ with even (odd) number of excitations, the number of dots on each horizontal line
must be even (odd). For any syndrome $s$ the number of dots with $\alpha=1\modn{4}$ and the number of dots
with $\alpha=-1\modn{4}$ must be even. }
\label{fig:topcharges}
\end{figure}

Consider as an example a quadrupole of type $t=[111]$ that consists of four excitations located at the sites
\be
\label{singleQ}
u_0-\hat{x},
\quad
u_0-\hat{y},
\quad
u_0-\hat{z},
\quad
u_0+\hat{x}+\hat{y}+\hat{z}
\ee
for some site $u_0\in \Lambda_{even}$, see Lemma~\ref{lemma:Tstring}.
Let $s$ be the corresponding syndrome. Let $\alpha$ be the odd integer such that $u_0\in \Sigma_{t,\alpha-1}$. Then $\theta_{t,\beta}(s)=1$ iff $\beta=\alpha \pm 2$.
Consider now some other body-diagonal $t'\ne t$.
Let $\alpha'$ be the odd integer such that $u_0\in \Sigma_{t',\alpha'-1}$.
Then one can easily check that two of the four sites in Eq.~(\ref{singleQ})
belong to the plane $\Sigma_{t',\alpha'-2}$ while the other two sites
belong to the plane $\Sigma_{t',\alpha'}$. Hence $\theta_{t',\beta}(s)=0$
for all $t'\ne t$ and all $\beta$. Similar arguments can be
made to quadrupoles of all other types, see Lemma~\ref{lemma:Tstring}.
We conclude that for any odd integer $\alpha$
and for any body-diagonal $t$ there exists a quadrupole which has only $\theta_{t,\alpha-2}$ and $\theta_{t,\alpha+2}$ as non-zero topological charges, see Fig.~\ref{fig:topcharges}.

Similarly we can determine the configuration of topological charges, Eq.~(\ref{theta}),
corresponding to a single dipole of type $h$. Consider as an example a dipole of type $h=[110]$
(that can be created by a rigid string lying in the $[001]$-plane)  that consists of two excitations located at the sites $u_0-\hat{x}$ and $u_0-\hat{y}$, see Section~\ref{subs:Hformal}. These two sites belong to the same plane $\Sigma_{t,\alpha}$ for the body-diagonals $t=[111]$ and $t=[\bar{1}\bar{1}1]$. Hence the corresponding topological charges $\theta_{t,\beta}$ are zero for all $\beta$.
On the other hand, the sites $u_0-\hat{x}$ and $u_0-\hat{y}$ belong to two different
planes $\Sigma_{t,\alpha\pm 1}$ if $t=[1\bar{1}\bar{1}]$ or $t=[\bar{1}1\bar{1}]$, where $\alpha$ is the even integer such that $u_0\in \Sigma_{t,\alpha}$. Hence we have four non-zero topological charges $\theta_{t,\alpha \pm 1}$ where $t=[1\bar{1}\bar{1}]$ or $t=[\bar{1}1\bar{1}]$, see Fig.~\ref{fig:topcharges}. Note that these are the two body-diagonals which lie in the plane orthogonal to $h$. Similar arguments can be applied to dipoles of all other types.

Now we are ready to prove that any syndrome $s$ with an even number of
excitations can be decomposed into dipoles and quadrupoles. First, we note that for any such syndrome
\be
\label{constr1}
\sum_{\alpha} \theta_{t,\alpha}(s)=\sum_u s(u)=0.
\ee
Secondly, we have an identity
\be
\label{constr2}
\sum_t \; \; \sum_{\alpha = 1\modn{4}} \; \theta_{t,\alpha}(s) =
\sum_t \; \; \sum_{\alpha = -1\modn{4}} \; \theta_{t,\alpha}(s)=
0,
\ee
where $t$ runs over the set of body-diagonals, see Eq.~(\ref{body-diag}).
Indeed, let $u=(i,j,k)\in \Lambda_{odd}$, $t=[t_xt_yt_z]$ be some body-diagonal,
and $u\cdot t\equiv it_x+jt_y+kt_z$.
Since $i+j+k$ is odd, the inner product $u\cdot t$ is odd for all body-diagonals $t$.
One can easily check that $\sum_t u\cdot t=0$, where the sum runs over
all body-diagonals $t$. Hence the number of $t$ such that $u\cdot t=1\pmod{4}$ must
be even. Analogously, the number of $t$ such that $u\cdot t=-1\pmod{4}$ must
be even.  It implies Eq.~(\ref{constr2}).

Let us represent $s$ by a diagram as shown on Fig.~\ref{fig:topcharges}
such that each non-zero charge $\theta_{t,\alpha}$ is represented by a dot.
Combining $s$ with a single quadrupole we can shift any dot along the horizontal axis by distance $4$,
that is, $(t,\alpha) \to (t,\alpha\pm 4)$. It allows us to concentrate all dots in the pair of columns
$\alpha =\pm 1$.
Since $\theta_{t,-1}(s)+\theta_{t,1}(s)=0$, see Eq.~(\ref{constr1}), we are
left only with four charges
\[
\theta_t(s)\equiv \theta_{t,1}(s)=\theta_{t,-1}(s)
\]
labeled by body-diagonals $t$, see Eq.~(\ref{body-diag}).
From Eq.~(\ref{constr2}) we infer that
\[
\sum_t \theta_t(s)=0,
\]
that is, the number of non-zero charges $\theta_t(s)$ must be even.
Any such syndrome $s$ can be decomposed into a single dipole or a pair of dipoles.
Indeed, suppose $\theta_t(s)=1$ for some pair of body-diagonals $t_1,t_2$.
Then $s$ can be created by a single dipole of type $h$, where $h$ is the face-diagonal
orthogonal to $t_1$ and $t_2$, see Fig.~\ref{fig:topcharges} for an example.
In the remaining case, $\theta_t(s)=1$ for all body-diagonals $t$.
Then $s$ can be created by a pair of dipoles of types $[110]$ and $[1\bar{1}0]$.
Summarizing, any syndrome $s$ with even number of excitations can be decomposed
(in the non-unique way) into a combination of quadrupoles and dipoles.

\subsection{Rigid strings with dislocations}
\label{subs:misaligned}
Let us now consider a pair of parallel rigid strings $\gamma_1,\gamma_2$
that are shifted with respect to each other as shown on Fig.~\ref{fig:dislocation}.
We would like to connect the two strings together
obtaining a rigid string with a `dislocation' and
explore what topological charges describe such a dislocation.

Let $W(\gamma_1)$, $W(\gamma_2)$ be the  string operators corresponding to $\gamma_1,\gamma_2$,
see Section~\ref{subs:Hformal}. Recall that $W(\gamma_1)$ creates a dipole near each end-point
of $\gamma_1$. Assume without loss of generality that the dipole located near the dislocation
occupies a pair of sites
\[
u_1=(-1,2,0) \quad \mbox{and} \quad v_1=(0,1,0).
\]
Similarly, $W(\gamma_2)$ creates a dipole that occupies a pair of sites
\[
u_2=(1,0,0) \quad \mbox{and} \quad v_2=(2,-1,0)
\]
and another dipole on the opposite end-point of $\gamma_2$.
Hence an operator  $W(\gamma_1\gamma_2)=W(\gamma_1)W(\gamma_2)$
creates four excitations at sites $u_1,v_1,u_2,v_2$.
Let $s$ be the corresponding syndrome. One can easily check that the only non-zero
topological charges $\theta_{t,\alpha}(s)$ correspond to body-diagonals
$t=[\bar{1}1\bar{1}]$, $t=[1\bar{1}\bar{1}]$ and $\alpha\in \{-3,-1,1,3\}$, see
the diagram shown on Fig.~\ref{fig:qqqq}.
One can easily check that this diagram is equivalent to a pair of quadrupoles
of type $t=[1\bar{1}\bar{1}]$ and a pair of quadrupoles of type $t=[\bar{1}1\bar{1}]$,
see the diagram of a single quadrupole shown on Fig.~\ref{fig:topcharges}.
For any $t$ as above one has one quadrupole in the $[t]$-bilayer $\Sigma_{t,-2}\cup \Sigma_{t,0}$
and one quadrupole in the $[t]$-bilayer $\Sigma_{t,0}\cup \Sigma_{t,2}$.
We shall see in Section~\ref{sec:encoding} that for a finite lattice with periodic boundary
conditions satisfying $g=1$ one has only one $[t]$-bilayer for each body-diagonal $t$
(different planes $\Sigma_{t,\alpha}$ must be identified due to the periodic boundary conditions).
Hence for a finite lattice each pair of quadrupoles of type $t$ can be
paired up with flexible strings and annihilated by the corresponding
string operators. Note that  if the lattice dimensions are pairwise co-prime,
a rigid string aligned with a face-diagonal cannot be closed directly since the corresponding diagonal
fully fills up a two-dimensional plane before it gets closed.
Therefore creating a dislocation
(or several dislocations) might be the only possible way to construct a closed rigid string for such lattice geometry.

\begin{figure}[htb]
\centerline{
\includegraphics[height=4cm]{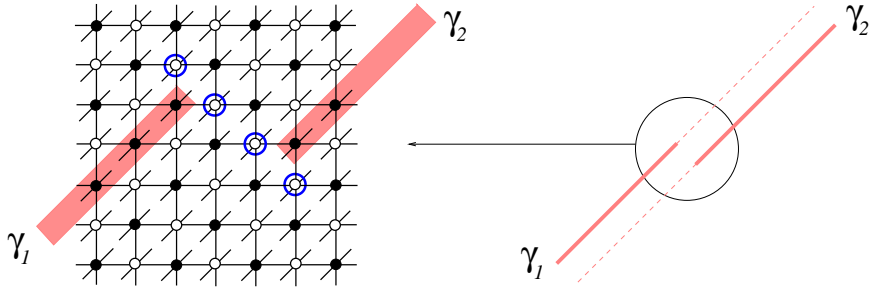}}
\caption{A pair of parallel shifted  rigid strings $\gamma_1$ and $\gamma_2$ of type $[110]$ can be connected
into a single rigid string $\gamma_1\gamma_2$ with a dislocation.
Double circles indicate locations of excitations associated with the dislocation. }
\label{fig:dislocation}
\end{figure}
\begin{figure}[htb]
\centerline{
\includegraphics[height=3cm]{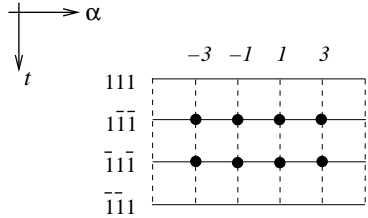}}
\caption{The diagram shows topological charges associated with a dislocation on a rigid string,
see Fig.~\ref{fig:dislocation}.
These charges are equivalent to a pair of quadrupoles of type $t=[1\bar{1}\bar{1}]$
and a pair of quadrupoles of type $t=[\bar{1}1\bar{1}]$ located in adjacent
$[t]$-bilayers.}
\label{fig:qqqq}
\end{figure}

\section{Encoding of a qubit using closed rigid strings}
\label{sec:encoding}

Let us remind the reader of some standard notations pertaining to stabilizer codes, see e.~g.~\cite{BT:mem}. Let $\calP=\la X_1,Z_1,\ldots,X_n,Z_n\ra$ be the $n$-qubit Pauli group.
Given a Pauli operator $P$, its weight $|P|$ is defined as the number of qubits
on which $P$ acts non-trivially (as $X$, $Y$, or $Z$).
A stabilizer code can be defined by an abelian subgroup
${\calS} \subseteq  \calP$  such that $-I\notin \calS$.
The corresponding codespace $\calL$ is spanned by $n$ qubit states invariant under the action of
any element of $\calS$.
Pauli operators commuting with every element of $\calS$
form the centralizer of $\calS$ denoted as
$\calC({\calS})=\{P \in \calP\, : \, PQ=QP \; \forall Q \in \calS\}$.
The elements of the centralizer preserve the codespace $\calL$.
For a stabilizer code encoding $k$ logical qubits the centralizer can be represented
as $\calC(\calS)=\la \calS,\bar{X}_1,\bar{Z}_1,\ldots,\bar{X}_k,\bar{Z}_k\ra$,
where $\bar{X}_i$ and $\bar{Z}_i$ are logical Pauli operators on the $i$-th encoded
qubit.  In this case $\dim{\calL}=2^k$. The distance of the code ${\calS}$ is defined as
minimum weight of a logical operator,
$d=\min_{P \in \calC(\calS)\backslash \calS} |P|$.

If the code encodes more than one qubit, some of them may be protected better than the others.
In this case it may be advantageous to use only a subset of $k'=k-m$ `good' logical qubits to encode information.
We then consider the gauge group $\calG=\la \calS, \bar{X}_1,\bar{Z}_1,\ldots, \bar{X}_m,\bar{Z}_m\ra$
generated by stabilizers and the logical operators acting on the remaining $m$ `bad' logical qubits.
 Pauli operators acting non-trivial
on the good  logical qubits are elements of $\calC(\calS) \backslash \calG$.
Hence the distance of the code in which quantum information is encoded into the `good' subsystem is
 $d_{\calG}=\min_{P \in \calC(\calS)\backslash \calG} |P|$.

In the next subsections, we assume that the lattice has periodic boundary conditions,
\[
\Lambda=\ZZ_{L_x}\times \ZZ_{L_y} \times \ZZ_{L_z},
\]
where $L_\alpha=2p_\alpha$ for some integers $p_x,p_y,p_z$. In addition, we shall focus on the
special case when  $p_x,p_y,p_z$ are odd and pairwise co-prime,
\be
\label{coprime}
\mathrm{gcd}(p_x,p_y)=\mathrm{gcd}(p_y,p_z)=\mathrm{gcd}(p_z,p_x)=1,
\ee
\be
\label{odd}
p_\alpha \modtwo =1, \quad \alpha=x,y,z.
\ee
The purpose of the co-primality constraint Eq.~(\ref{coprime})
is to take advantage of the rigid strings present in the model to maximize the code distance,
see Section~\ref{subs:rclosed}. The constraint Eq.~(\ref{odd}) is introduced to simplify
classification of logical operators and is not essential.

\subsection{Half-filled membrane operators}
\label{subs:half}

In this section we construct a complete set of logical operators for the code $\calS$.
Under conditions Eq.~(\ref{coprime}), Theorem~\ref{thm:1} implies that the code $\calS$ has $k=4$ logical qubits. One can choose the logical Pauli operators on the four encoded qubits using
{\em half-filled} membrane operators which are supported on four disjoint sublattices of $\Lambda_{even}$. Let $a,b,c \in \{0,1\}$ with $a+b+c \modtwo=0$ label the four qubits. For a qubit labeled by $abc$, we can find the logical operators
\bea
\label{logicals}
\bar{\sigma}^x_{abc}& \equiv &  \prod_{j=b \modtwo}\,\prod_{k=c \modtwo} X_{a,j,k}, \nonumber \\
\bar{\sigma}^y_{abc}& \equiv & \prod_{i=a \modtwo} \, \prod_{k=c \modtwo} Y_{i,b,k}, \nonumber \\
\bar{\sigma}^z_{abc}& \equiv & \prod_{i=a \modtwo} \,\prod_{j=b \modtwo} Z_{i,j,c}.
\eea
Note that the logical operators $\bar{\sigma}^{\alpha}_{abc}$ are supported on a sublattice $\Lambda_{abc}$ for which the $x,y$ and $z$ coordinates have parity $a,b$ and $c$ respectively.
For example, the four logical $\bar{\sigma}^x_{abc}$ operators
 lie in one of two $[100]$-planes characterized by fixed $x$-coordinate
 $i=a=0$ or $i=a=1$. For such fixed $[100]$-plane, the choice of $b$ (and thus $c$) determines
 the parity of the $y$-coordinate $j$ and $z$ coordinate $k$,
  see Fig.~\ref{fig:half-filled}. A simple inspection shows that the
   operators $\bar{\sigma}^{\alpha}_{abc}$ commute with
  all generators $S_v$ and thus they are elements of the centralizer $\calC(\calS)$.

\begin{figure}[htb]
\centerline{
\includegraphics[height=5cm]{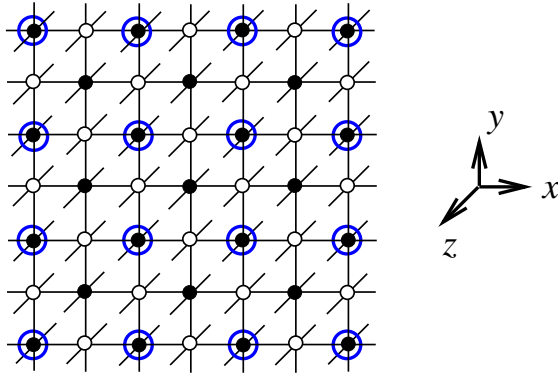}}
\caption{(Color Online) A half-filled membrane operator $\bar{\sigma}^z$ lying in some $[001]$-plane
is the product of Pauli $Z$'s  over all qubits which have a fixed parity of the $x$- and $y$-coordinates
(the qubits indicated by double circles).}
\label{fig:half-filled}
\end{figure}

It is straightforward to see that any pair of logical operators $\bar{\sigma}^{\alpha}_{abc}$ and
$\bar{\sigma}^{\beta}_{a'b'c'}$ commute when $abc \neq a'b'c'$ since they have disjoint supports.
Consider the commutation between, say, $\bar{\sigma}^x_{abc}$ and $\bar{\sigma}^y_{abc}$ and
 notice that they overlap on $p_z$ qubits. Since $p_z$ is odd, Eq.~(\ref{odd}), these logical operators thus anti-commute on an odd number of qubits, hence they anti-commute. The choice for odd $p_x$ and $p_y$ similarly ensures the proper commutation rules between $\bar{\sigma}^x_{abc}$ and $\bar{\sigma}^z_{abc}$ etc.

It is important to note that one can translate half-filled membrane operators by multiplication with the stabilizer generators. For example, one can shift the logical operator $\bar{\sigma}^x_{000}$ to the parallel $[100]$-plane $(i,j,k)=(2, j=0 \modtwo,k=0\modtwo)$ by multiplying it with stabilizer generators at the sites $(i,j,k)=(1,0 \modtwo,0 \modtwo)$. So we can translate any half-filled membrane operator lying in a plane orthogonal to the basis
vector $\hat{\alpha}$  by any vector $2 k \hat{\alpha}$ for integer $k$ and get an equivalent logical operator.

Although the logical Pauli operators introduced in Eq.~(\ref{logicals}) have a membrane-like geometry, it does not mean that any logical operator of the code $\calS$ has membrane-like geometry as well. Indeed, it may be possible to transform the half-filled plane operators or some combination of them into a logical operator with a smaller support (e.g. string-like) by multiplying them with the stabilizers. In fact, we will see in the next section that some combination of these logical operators are obtainable as closed strings winding around the lattice,  see Lemma \ref{lemma:111}.

\subsection{Logical operators associated with closed flexible strings}
\label{subs:fclosed}

Recall that flexible strings are confined to $[t]$-bilayers $\Sigma_{t,\alpha-1}\cup \Sigma_{t,\alpha+1}$,
where $t$ is a body-diagonal and  $\Sigma_{t,\alpha}$ is a $[t]$-plane that includes all sites $u=(i,j,k)$ satisfying $t_x i+ t_y j+t_zk=\alpha$. Note that for periodic boundary conditions and lattice dimensions satisfying Eq.~(\ref{coprime})
the sum $t_x i+ t_y j+t_zk$ is defined only modulo $\mathrm{gcd}(L_x,L_y,L_z)=2\mathrm{gcd}(p_x,p_y,p_z)=2g=2$.
Hence the entire 3D lattice can be considered as a single $[t]$-bilayer embedded into the $3$-torus
and folded into a closed surface.
One can easily check that any such bilayer has topology of the two-dimensional torus\footnote{For example, one can
check that any  $[t]$-bilayer is a closed orientable surface with Euler characteristic $\chi=0$ since it admits  covering by the regular hexagonal lattice.
It follows that this surface has  genus $g=1$, that is, it has the  topology of the $2$-torus.}.
It follows that flexible strings can be closed in a topologically non-trivial way by winding around the torus.
We would like to characterize the corresponding string operators in terms of the logical operators Eq.~(\ref{logicals}).
Let us define the following combination of the logical operators:
\be
\label{Zs12}
\bar{Z}_1= \bar{\sigma}^{x}_{000} \bar{\sigma}^{x}_{011}\bar{\sigma}^{x}_{101}\bar{\sigma}^{x}_{110}, \quad \bar{Z}_2  = \bar{\sigma}^{z}_{000} \bar{\sigma}^{z}_{011}\bar{\sigma}^{z}_{101}\bar{\sigma}^{z}_{110}.
\ee
In this section we will show that $\bar{Z}_1$, $\bar{Z}_2$, and $\bar{Z}_1\bar{Z}_2$
are the only logical operators that correspond to closed flexible strings
winding around the lattice. On the other hand, elements of $\calS$ correspond to contractible
flexible loops, see Section~\ref{subs:Texample}. Hence the subgroup of $\calC(\calS)$ generated by
closed (trivial or non-trivial) flexible strings can be described as
\be
\label{gauge1}
\calG'=\la \calS,\bar{Z}_1,\bar{Z}_2\ra.
\ee
\begin{lemma}
\label{lemma:111}
The group $\calG'$ is the group generated by
all string operators $W(\gamma)$ associated with closed flexible
 strings $\gamma$.
\end{lemma}
\begin{proof}
First, recall that stabilizer generators can be regarded as  homologically trivial flexible loops,
see Section~\ref{subs:Texample}. In addition, it is clear that $W(\gamma) \in \calC(\calS)$ for any closed
flexible  string $\gamma$.
Let us first prove that
\be
\label{OOrules}
W(\gamma) \bar{\sigma}^{\alpha}_{abc} =(-1)^{w_\alpha} \bar{\sigma}^{\alpha}_{abc} W(\gamma),
\ee
where $(w_x,w_y,w_z)\in H_1(S^1\times S^1\times S^1,\ZZ_2)$ represents the homological class of $\gamma$. In other words,
$w_\alpha$ is the parity of the number of times the $\alpha=x,y,z$ coordinate of the string $\gamma$ winds around $\ZZ_{L_{\alpha}}$. Using the lattice symmetries, it suffices to consider the case when $\alpha=z$ and $\gamma$ lies in the $[111]$-bilayer.
Let $M_{abc}$ be the support of $\bar{\sigma}^z_{abc}$ (a half-filled membrane) and
$\pi$ be the $[001]$-plane containing $M_{abc}$.
 Let $\gamma=(u_0,u_1,\ldots,u_t=u_0)$. Recall that $\gamma$ is a sequence of steps $\lambda^x,\lambda^y,\lambda^z$ with alternating signs, see the definition of $\lambda^{\alpha}$ in Eq.~(\ref{lambdas}) in Section~\ref{subs:Texample}. Without loss of generality, the string never makes two consecutive steps that undo each other such that $u_{j+2}=u_j$ for some $j$.
We shall say that a site $u_j$ is a crossing point iff $u_j\in \pi$ and $u_{j+1}=u_j\pm \lambda^z$ (and hence both $u_j$, $u_{j+1}$ belong to the plane $\pi$). Note that exactly one of the sites $u_j,u_{j+1}$ belongs to the half-filled membrane $M_{abc}$. Since the steps
$(u_{j-1},u_j)$ and $(u_{j+1},u_{j+2})$ must be $\pm \lambda^x$ or $\pm \lambda^y$, the string operator $W(\gamma)$ anti-commutes with $\bar{\sigma}^z_{abc}$ at exactly one of the sites $u_j,u_{j+1}$. Taking into account that the sites $u_{j-1}$ and $u_{j+1}$ are on the opposite sides
of $\pi$, we infer that the number of sites at which $W(\gamma)$ anti-commutes with
$\bar{\sigma}^z_{abc}$ and the winding number $w_z$ have the same parity.

Since each pair of consecutive steps $u\rightarrow u+\lambda^\alpha-\lambda^\beta$
preserves the sum of the coordinates $i+j+k$, the winding numbers
of any closed flexible  string must obey the constraint
\be
\label{wnc}
w_x + w_y+ w_z=0 \modtwo.
\ee
If all $w_\alpha$ are even, then Eq.~(\ref{OOrules}) implies that $W(\gamma)$ commutes with all logical operators, that is, $W(\gamma)\in \calS$ and hence $W(\gamma)$ is a trivial loop.
When $W(\gamma)\in \calC(\calS) \backslash \calS$ for some closed flexible  string $\gamma$, exactly two of the numbers $w_x,w_y,w_z$ must be odd. Consider for example the case when $w_x,w_y$ are odd and $w_z$ is even. Combining Eqs.~(\ref{OOrules},\ref{wnc}) we infer that
$W(\gamma)$ anti-commutes with all logical operators $\bar{\sigma}^x_{abc}$, $\bar{\sigma}^y_{abc}$ and commutes with all logical operators $\bar{\sigma}^z_{abc}$. There is only one operator in $\calC(\calS) \backslash \calS$ which has this property and this is $\bar{Z}_2=\prod_{abc} \bar{\sigma}^z_{abc}$. Similarly, when $w_x$ is even and $w_y,w_z$ is odd we infer that $W(\gamma)$
should be $\bar{Z}_1=\prod_{abc} \bar{\sigma}^x_{abc}$. Thus the group generated by all closed flexible strings $W(\gamma)$ equals the group $\la \calS,\bar{Z}_1,\bar{Z}_2\ra$. Note that the logical action of non-trivially closed flexible strings is only determined by a topological property, namely their winding numbers.
\end{proof}

One may ask how the minimum weight of nontrivially closed flexible strings scales as a function of $n$, where $n$ is the total number of qubits.  Let $\gamma_{min}$ be the shortest closed flexible string which is topologically non-trivial (i.e. non-contractible to a point). As was mentioned above, one can consider the entire 3D lattice $\Lambda$ as a single $[t]$-bilayer $\Sigma_t$
associated with some body-diagonal $t$ which is embedded into the $3$-torus and folded into a $2$-torus.
Note that $\Sigma_t$ has area of order $n$ since it contains $n$ qubits covering the bilayer with density of order $1$.
Then Loewner's torus inequality implies that the shortest non-contractible loop on the torus $\Sigma_t$ has length $|\gamma_{min}|\le c\sqrt{n}$ for some constant $c$.
 It means that the minimum weight of non-trivial closed flexible strings cannot grow faster than $c\sqrt{n}$.
In fact, one can show that for a particular sequence of lattice dimensions $p_x,p_y,p_z=\Theta(n^\frac13)$
 satisfying Eqs.~(\ref{coprime},\ref{odd})
 we also have a matching lower bound $|\gamma_{min}|=\Omega(\sqrt{n})$, see~\cite{leemhuis:mthesis}.

\subsection{Logical operators associated with closed rigid strings}
\label{subs:rclosed}

Consider as an example a rigid string associated with the face-diagonal $h=[110]$, that is, a sequence of sites $\gamma=(u_0,u_1,\ldots,u_m)$ where
$u_{i+1}=u_i+h$ for all $i$, see Definition~\ref{dfn:Hstring}.
In order to get a closed string, $u_m=u_0$, one must have
$m=0\modn{L_x}$ and $m=0\modn{L_y}$. The smallest positive $m$ satisfying these conditions
is the least common multiple of $L_x$ and $L_y$, that is,
\be
m=\frac{L_x L_y}{\mathrm{gcd}(L_x,L_y)} = \frac{4p_x p_y}{2\mathrm{gcd}(p_x,p_y)} =2p_xp_y.
\ee
Thus the string $\gamma$ has to wind $p_x$ times around the $y$-axis ($p_y$ times around the $x$-axis) before it gets closed. Such a string $\gamma$ completely fills up the $xy$-plane that contains the site $u_0$, that is, $\gamma$ is actually a closed two-dimensional membrane.
The corresponding string operator $W(\gamma)$ can be expressed as
$W(\gamma)=\bar{\sigma}^z_{000}\bar{\sigma}^z_{110}$ or $W(\gamma)=\bar{\sigma}^z_{101}\bar{\sigma}^z_{011}$
depending on whether $\gamma$ lies in $[001]$-plane with even and odd $z$-coordinate.
This operator has support on a {\em fully-filled membrane} since it involves all qubits lying in
some $[001]$-plane.
One can similarly construct logical operators associated with other closed rigid strings
$\gamma$. Such operators generate a subgroup
\be
\label{gauge2}
\calG''= \la  \bar{\sigma}^z_{000}\bar{\sigma}^z_{110},\;
\bar{\sigma}^z_{101}\bar{\sigma}^z_{011}, \;
\bar{\sigma}^y_{000}\bar{\sigma}^y_{101},\;
\bar{\sigma}^y_{011}\bar{\sigma}^y_{110},\;
\bar{\sigma}^x_{000}\bar{\sigma}^x_{011},\;
\bar{\sigma}^x_{110}\bar{\sigma}^x_{101}\ra.
\ee
By comparing Eq.~(\ref{gauge1}) and Eq.~(\ref{gauge2}) we conclude that string operators
associated with closed rigid and closed flexible strings pairwise commute, that is,
$PQ=QP$  for all $P\in \calG'$ and $Q\in \calG''$.
It allows us to choose a subsystem encoding with two logical qubits and two gauge qubits
such that  string operators associated with closed flexible  strings act only on the gauge qubits.
For example, let us complement $\bar{Z}_1,\bar{Z}_2$ defined in Eq.~(\ref{Zs12}) to a complete
set of logical operators as
\[
\ba{rclrcl}
\bar{X}_1 &=& \bar{\sigma}^{z}_{000} \bar{\sigma}^{z}_{011}\bar{\sigma}^{z}_{110}
&
\bar{Z}_1&=& \bar{\sigma}^{x}_{000} \bar{\sigma}^{x}_{011}\bar{\sigma}^{x}_{101}\bar{\sigma}^{x}_{110} \\
 \bar{X}_2&=& \bar{\sigma}^{x}_{101} &
\bar{Z}_2  &=& \bar{\sigma}^{z}_{000} \bar{\sigma}^{z}_{011}\bar{\sigma}^{z}_{101}\bar{\sigma}^{z}_{110}\\
\bar{X}_3 &=& \bar{\sigma}^x_{000} \bar{\sigma}^x_{011} &
\bar{Z}_3 &=& \bar{\sigma}^z_{000} \bar{\sigma}^z_{110} \\
\bar{X}_4 &=& \bar{\sigma}^x_{000} \bar{\sigma}^x_{110} &
\bar{Z}_4 &=& \bar{\sigma}^z_{000} \bar{\sigma}^z_{011}\\
\ea
\]
One can check that $\bar{X}_i,\bar{Z}_i$ obey the commutation rules of the Pauli operators
on four qubits.  We can now regard qubits $1,2$ as gauge qubits (which do not encode any information)
while the remaining qubits $3,4$ are the logical qubits of the code.
The corresponding subsystem code has a gauge group
\[
\calG=\la \calS,\bar{Z}_1,\bar{Z}_2,\bar{X}_1,\bar{X}_2\ra.
\]
Lemma~\ref{lemma:111} implies that string operators associated with closed flexible strings
act only on the gauge qubits.
Interestingly, the logical operators $\bar{X}_3$ and $\bar{Z}_3$ correspond to fully-filled membranes, while the logical operators $\bar{X}_4$ and $\bar{Z}_4$ correspond to a pair
of half-filled membranes occupying two adjacent $[100]$- and $[001]$-planes respectively.


In order to find the distance $d_{\calG}$ of this subsystem code one has
to minimize the weight of the logical operators $\bar{X}_3,\bar{X}_4,\bar{Z}_3,\bar{Z}_4$ by multiplying them with
stabilizers and gauge operators. This distance must obey the general upper bound $d_{\calG}=O(L^2)$, where
$L$ is the smallest of the lattice dimensions, see Theorem~1${}^*$ in~\cite{BT:mem}.
Computing the distance $d_{\calG}$  leads to a highly non-trivial optimization problem.
To illustrate the potential difficulties let us assume that the lattice dimensions obey one extra constraint
$p_y=p_x+2$. In this case one can construct a closed rigid string $\gamma$ lying in some $[001]$-plane
that winds around the lattice only once
by creating a single dislocation, see Fig.~\ref{fig:dislocation}. As was explained in Section~\ref{subs:misaligned},
the quadrupoles located in the neighborhood of the dislocations can be annihilated by pairing them up
with flexible strings $\delta$. The total length of these flexible strings is  $|\delta|=O(\sqrt{n})$, see Section~\ref{subs:fclosed}.
It allows us to construct a logical operator $W=W(\gamma)W(\delta)$
whose support has geometry of a string-net
--- a collection of rigid and flexible strings interconnected with each other. The total length of the string-net
$\gamma\cup \delta$ grows at most as $O(\sqrt{n})$. Note that the logical operator $W$
anti-commutes with the fully-filled membrane operator $\bar{X}_3$. Indeed,  we can assume that the
fully-filled membrane corresponding to $\bar{X}_3$ crosses the rigid string $\gamma$ at exactly one qubit
and hence $W$ anti-commutes with $\bar{X}_3$. On the other hand, $\bar{X}_3$ commutes with
the flexible string operators $W(\delta)$, see Eq.~(\ref{OOrules}), assuming that the end-points
of $\delta$ are sufficiently far from the membrane corresponding to $\bar{X}_3$.
It follows that $W$ anti-commutes with $\bar{X}_3$ and thus $W$ cannot be in the gauge group, that is,
$d_{\calG}\le |W|=O(\sqrt{n})$.

In the general case the minimum-length string-nets corresponding to closed rigid strings
are more difficult to analyze. Whether the distance obeys the upper bound
$d_{\calG}=O(\sqrt{n})$ for an arbitrary choice of the dimensions $p_x,p_y,p_z$
is an interesting open problem.

\section{Discussion and open problems}
\label{sec:open}


We would like to remark that our model suggests natural generalizations of models other than the surface code to the third dimension. Consider for example the qudit surface codes \cite{BB:qsc}
defined on a lattice with qudits with prime-dimension $d$ which realizes a 2D $\ZZ_d$ gauge theory. To get a three-dimensional extension, we can take the lattice $\Lambda=\Lambda_{even} \cup \Lambda_{odd}$ defined above and for every $u \in \Lambda_{odd}$ we define a generator
$S_u =X_{u-\hat{x}}\,  X^{-1}_{u+\hat{x}}\,  Y_{u-\hat{y}}\,  Y^{-1}_{u+\hat{y}} \, Z_{u-\hat{z}}\, Z^{-1}_{u+\hat{z}}+h.c.$ where $X$ and $Z$ are generalized Pauli-operators with
$X^d=Z^d=I$ and $X^{-1}=X^{\dagger}$ and $Z^a X^b=\exp(2\pi i a \cdot b/d) X^b Z^a$.
One can easily check that $S_u$ are mutually-commuting (use
 $[A \otimes B,B^{\dagger} \otimes A^{\dagger}]=0$ which holds for any generalized Pauli operators $A,B$). An open question is also how one can extend the quantum double models based on non-Abelian gauge groups \cite{kitaev:anyons} to the third dimension in a similar fashion.

A natural question is whether the studied 3D model exhibits thermal stability such that the
relaxation time of error corrected logical operators grows as a function of the lattice size.
We believe that similar to the 2D toric code, topological order in our model is lost at any non-zero temperature.
Consider for example a simplified noise model with only $Z$ errors. In this case one can consider the dynamics of
excitations created in any $[001]$-plane separately. In any given $[001]$-plane  one may expect a gas (whose density depend on the temperature) of dipole strings whose ends freely oscillate along diagonal lines. It is favorable energetically for dipole strings to combine to former longer oscillating strings which diagonally wind over the lattice and eventually can become the logical operator $\bar{Z}_3$ (the fully-filled membrane lying in the $[001]$-plane). In other words, one expects that the thermal behavior of some of these excitations is like those of a 1D Ising model which has been folded into a  2D plane.

\section*{Acknowledgements}
BMT and SB acknowledge support by the DARPA QUEST program under contract number HR0011-09-C-0047. BMT and SB would like to thank Dave Bacon for first introducing this spin model to them and they acknowledge insightful discussions concerning properties of the model with Panos Aliferis, Dave Bacon and Alexei Kitaev. BL acknowledges the financial support and the warm hospitality from IBM Research and its employees.


\appendix

\section{Ground state degeneracy}
\label{sec:appA}

In this Appendix we prove Theorem~\ref{thm:1}. It suffices
to show that the stabilizer group $\calS$ has $n-k$ independent generators with $k=4g$
and $\calS$ does not contain $-I$.

Let us begin by introducing some notations.
Let $M$ be any finite set and $\calB=\{0,1\}$.  A
set of bit strings of length $|M|$ in which bits are labeled by elements of $M$
will be denoted $\calB(M)$.
Given a string $s\in \calB(M)$ and $u\in M$ let $s_u\in \calB$ be the bit corresponding to $u$.
For any $s,t\in \calB(M)$ let $s\oplus t\in \calB(M)$ be the bit-wise XOR of $s$ and $t$.
A subset $\calC\subseteq \calB(M)$ is called a (binary) linear subspace
iff $s\oplus t\in \calC$ for any $s,t\in \calC$.
The dimension of a linear subspace $\calC$ defined over the binary field will be denoted $\dim{\calC}$.
Note that $|\calC|=2^{\dim{\calC}}$ for any subspace $\calC$.
A subspace $\calC$ that includes only constant strings $s=(00\ldots 0)$  and $s=(11\ldots 1)$
will be referred to as a repetition code on $M$.

Recall that $g=\mathrm{gcd}(p_x,p_y,p_z)$. We shall begin in Section~\ref{subs:g=1} by proving Theorem~\ref{thm:1}
in the special case
$g=1$ since in this case the proof is much simpler.
It will allow us to introduce necessary machinery for the general proof, see Sections~\ref{subs:k},\ref{subs:-I}.

\subsection{The special case $g=1$}
\label{subs:g=1}

Define a linear subspace $\calC=\calC(p_x,p_y,p_z)\subseteq \calB(\Lambda_{odd})$  describing linear dependencies among the generators $S_u$,
\[
\calC(p_x,p_y,p_z)=\{ t\in \calB(\Lambda_{odd})\, : \, \prod_{u\in \Lambda_{odd}} \, S_u^{t_u} =\pm  I \}.
\]
By abuse of notations we shall omit the dependence on $p_x,p_y,p_z$ whenever it is not important.
Clearly the stabilizer group $\calS$ has $n-\dim{\calC}$ independent generators
and thus the number of logical qubits is
\[
k= \dim{\calC}
\]
Using the explicit form of generators $S_u$ we infer that $t\in \calC$
iff $t$ satisfies parity checks
\bea
t_{u-\hat{x}} \oplus t_{u+\hat{x}} \oplus t_{u-\hat{y}} \oplus t_{u+\hat{y}} &=& 0 \label{checksXY} \\
t_{u-\hat{x}} \oplus t_{u+\hat{x}} \oplus t_{u-\hat{z}} \oplus t_{u+\hat{z}} &=& 0 \label{checksXZ} \\
t_{u-\hat{y}} \oplus t_{u+\hat{y}} \oplus t_{u-\hat{z}} \oplus t_{u+\hat{z}} &=&0 \label{checksYZ}
\eea
for every site $u\in \Lambda_{even}$.
\begin{lemma}
\label{lemma:Code}
For any  $t\in \calC$  one has
\be
\label{period2g}
t_{u}=t_{u+2g\hat{x}}=t_{u+2g\hat{y}}=t_{u+2g\hat{z}} \quad \mbox{for all $u\in \Lambda_{odd}$}.
\ee
If $g=1$ then the converse is true, that is, any string $t\in \calB(\Lambda_{odd})$ satisfying Eq.~(\ref{period2g})
belongs to $\calC$.
\end{lemma}
In the special case $g=1$
Lemma~\ref{lemma:Code} implies that the subspace $\calC$ can be regarded as a product of four
repetition codes defined on the four non-overlapping sublattices
$\Lambda_{100}$, $\Lambda_{010}$, $\Lambda_{001}$, and $\Lambda_{111}$, where
\[
\Lambda_{abc}=\{ (i,j,k)\in \Lambda\, : \, i=a\modtwo, \quad j=b\modtwo, \quad k=c\modtwo\}.
\]
Since the product of generators $S_u$ over each of these sublattices gives $+I$ we conclude
that $\dim{\calC}=4$ and $-I\notin \calS$ thus proving Theorem~\ref{thm:1}.
In the rest of the section we prove Lemma~\ref{lemma:Code}.
\begin{proof}
By taking linear combinations of the parity checks Eqs.~(\ref{checksXY}-\ref{checksYZ}) one
easily gets
\bea
t_{u-h\hat{x}} \oplus t_{u+h\hat{x}} \oplus t_{u-h\hat{y}} \oplus t_{u+h\hat{y}} &=& 0 \label{checksXYh} \\
t_{u-h\hat{x}} \oplus t_{u+h\hat{x}} \oplus t_{u-h\hat{z}} \oplus t_{u+h\hat{z}} &=& 0 \label{checksXZh} \\
t_{u-h\hat{y}} \oplus t_{u+h\hat{y}} \oplus t_{u-h\hat{z}} \oplus t_{u+h\hat{z}} &=&0 \label{checksYZh}
\eea
for any non-negative integer $h$.
Let us try to choose $h$ such that the above parity checks become equivalent to the ones
in Eq.~(\ref{period2g}). We shall use the following well-known fact.
\begin{prop}[\bf Chinese Remainder Theorem]
\label{prop:CRT}
A system of equations
\bea
h&=& a_1 \modn{n_1} \nn \\
h&=& a_2 \modn{n_2} \nn
\eea
has a solution $h$ iff $a_1=a_2 \modn{n}$, where $n=\mathrm{gcd}(n_1,n_2)$.
\end{prop}
Let $g_{\alpha\beta}=\mathrm{gcd}(p_\alpha,p_\beta)$.
Note that $\mathrm{gcd}(L_\alpha,L_\beta)=2g_{\alpha\beta}$.
By Proposition~\ref{prop:CRT} we can
choose integers $h$, $h'$  such that
\[
\left\{ \ba{rcll} 2h&=& 0 &\modn{L_y} \\
2h&=& 2g_{xy} &\modn{L_x} \\ \ea \right. \quad \mbox{and} \quad
\left\{ \ba{rcll} 2h' &=& 0 &\modn{L_z} \\
2h' &=& 2g_{xz} &\modn{L_x} \\ \ea \right.
\]
Using these $h$ and $h'$ in resp. the parity checks Eq.~(\ref{checksXYh}) and Eq.~(\ref{checksXZh}) we get
\bea
t_u&=&t_{u+2g_{xy}\hat{x}}, \nn \\
t_u&=&t_{u+2g_{xz}\hat{x}},\nn
\eea
for all $u\in \Lambda_{odd}$.
By B$\acute{{\rm e}}$zout's identity, the smallest positive integer obtained as
an integer linear combination of $2g_{xy}$ and $2g_{xz}$ is equal to
$\mathrm{gcd}(2g_{xy},2g_{xz})=2\mathrm{gcd}(g_{xy},g_{xz})=2g$, that is, we get
$t_u=t_{u+2g\hat{x}}$ for all $u\in \Lambda_{odd}$.
The remaining parity checks in Eq.~(\ref{period2g}) are obtained analogously.
In the special case $g=1$ the parity checks Eq.~(\ref{period2g})
specify a product of four repetition codes defined on the non-overlapping sublattices
$\Lambda_{100}$, $\Lambda_{010}$, $\Lambda_{001}$, and $\Lambda_{111}$.
One can easily check that such repetition codes satisfy all parity checks
Eqs.~(\ref{checksXY}-\ref{checksYZ}). Hence in this case Eq.~(\ref{period2g})
is a necessary and sufficient condition for $t\in \calC$.
\end{proof}

\subsection{Computing the number of logical qubits for $g>1$}
\label{subs:k}

We can use Lemma~\ref{lemma:Code} to establish an isomorphism between
$\calC(p_x,p_y,p_z)$ and $\calC(g,g,g)$ as follows.
Let $\Lambda'=\ZZ_{2g}\times \ZZ_{2g}\times \ZZ_{2g}$.
Given a string $t\in \calB(\Lambda_{odd})$ obeying Eq.~(\ref{period2g}),
let $t'\in \calB(\Lambda'_{odd})$
be a string defined through $t'_{\alpha,\beta,\gamma}=t_{i,j,k}$, where
$\alpha=i\modn{2g}$, $\beta=j\modn{2g}$, and $\gamma=k\modn{2g}$.
Lemma~\ref{lemma:Code} tells us that
$t\in \calC(p_x,p_y,p_z)$ iff Eq.~(\ref{period2g}) holds and $t'\in \calC(g,g,g)$.
Hence in order to prove the identity $k=4g$ it suffices to show that
\be
\label{ggg1}
\dim{\calC(g,g,g)}=4g.
\ee
Let us begin by considering a two-dimensional version of this problem.
Consider a 2D square lattice $\Omega=\ZZ_{2g}\times \ZZ_{2g}$.
A site $u=(i,j)\in \Omega$ is called even (odd) iff $i+j$ is even (odd).
Let $\Omega_{even}$ and $\Omega_{odd}$ be the even and odd sublattices.
Define a linear subspace
\be
\label{2dCode}
\calC(g,g)=\{ t\in \calB(\Omega_{odd})\, : \, t_{i+1,j}\oplus t_{i-1,j} \oplus t_{i,j+1}\oplus t_{i,j-1}=0
\quad \mbox{for all $(i,j)\in \Omega_{even}$}\}.
\ee
We claim that
\be
\label{gg}
\dim{\calC(g,g)}=2g.
\ee
Indeed, considering $t=\{t_{i,j}\}$ as a binary matrix of size $2g\times 2g$ (with zero entries at all
even cells), one can put an arbitrary binary strings at the  rows $i=0,1$
since there are no parity checks that have support only in these rows. Using the parity checks
Eq.~(\ref{2dCode}) one can uniquely fill up the remaining rows $i=2,3,\ldots,2g-1$.
Furthermore, using the generalized parity checks Eq.~(\ref{checksXYh}) with $h=g$
one arrives at $t_{i+2g,j}=t_{i,j}$, that is, expressing the row $i=0$ in terms of the rows
$i=2g-2,2g-1$ one will always come back to the original binary string $\{t_{0,j}\}$.
Similarly, expressing the row $i=1$ in terms of the rows $i=2g-1,0$ one will always
come back to the original binary string $\{t_{1,j}\}$.   Hence $t\in \calC(g,g)$ is uniquely determined by $2g$ bits
$t_{0,j}$, $(0,j)\in \Omega_{odd}$, and  $t_{1,j}$, $(1,j)\in \Omega_{odd}$.
These  bits can assume arbitrary values. It proves Eq.~(\ref{gg}).
The above arguments also show that we can ignore the parity checks
centered at the first and the last rows, that is, $\calC(g,g)$ can be specified by the parity checks
\be
\label{2dCode'}
t_{i+1,j}= t_{i,j+1} \oplus t_{i,j-1} \oplus t_{i-1,j},
\quad j\in \ZZ_{2g}, \quad 1\le i\le 2g-2, \quad (i,j)\in \Omega_{even}.
\ee

\noindent
{\em Remark:}  For the later use we show examples of strings $t\in \calC(g,g)$ on Fig.~\ref{fig:CA}.
One can easily check that these strings and their translations generate the entire subspace
$\calC(g,g)$.

\begin{figure}[htb]
\centerline{
\includegraphics[height=4cm]{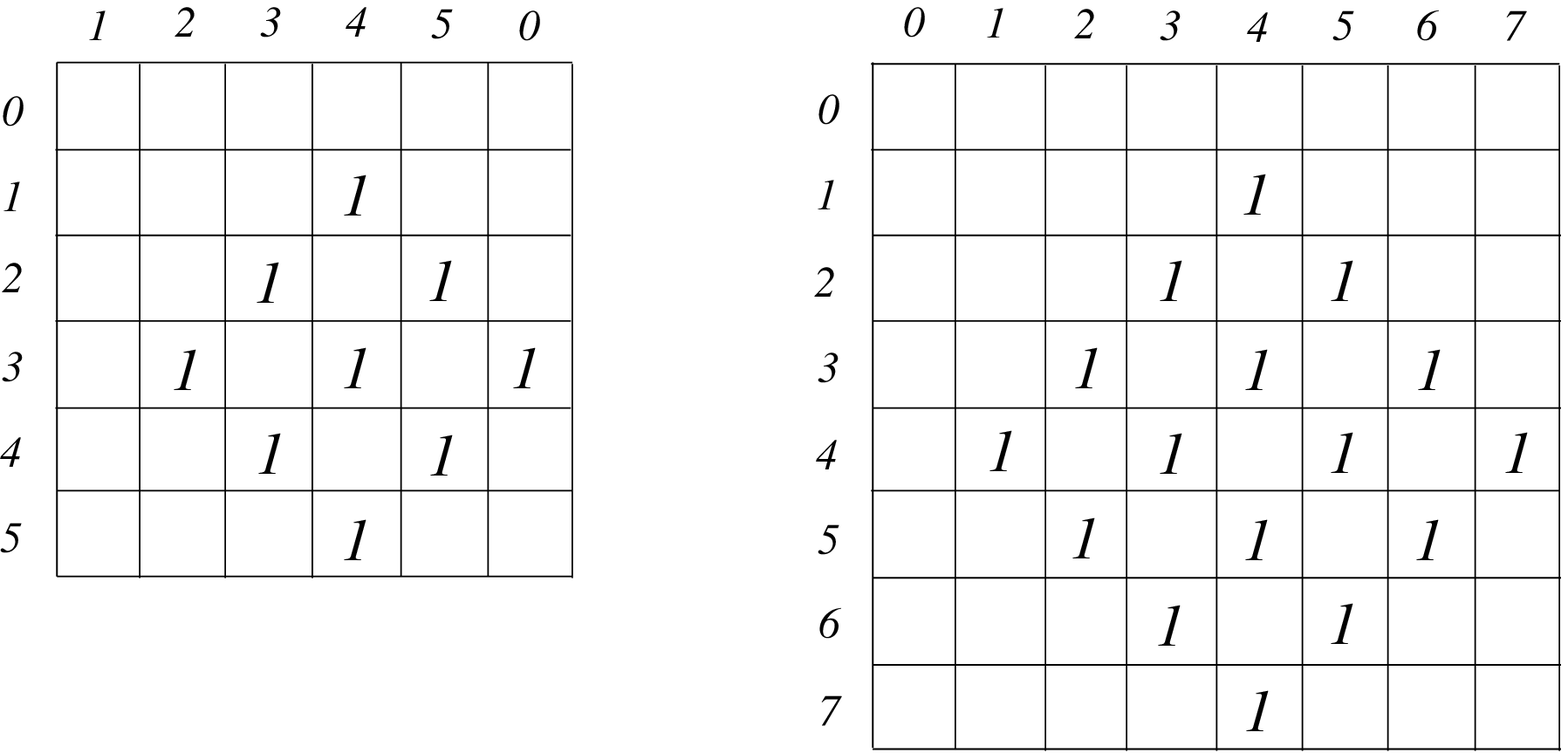}}
\caption{Examples of a string $t\in \calC(g,g)$ for $g=3,4$ (only non-zero bits are shown).
Note that these strings have a single non-zero bit in the first pair of rows.
It follows that the entire subspace $\calC(g,g)$ is spanned by
horizontal and vertical translations of the shown string.}
\label{fig:CA}
\end{figure}

Now we can easily prove Eq.~(\ref{ggg1}).  Consider any string $t\in \calC(g,g,g)$.
We note that for any site $u\in \Lambda_{even}$ the parity checks
Eqs.~(\ref{checksXY},\ref{checksXZ},\ref{checksYZ}) are linearly dependent, namely,
\[
(t_{u-\hat{x}} \oplus t_{u+\hat{x}} \oplus t_{u-\hat{y}} \oplus t_{u+\hat{y}})\oplus
(t_{u-\hat{x}} \oplus t_{u+\hat{x}} \oplus t_{u-\hat{z}} \oplus t_{u+\hat{z}}) \oplus
(t_{u-\hat{y}} \oplus t_{u+\hat{y}} \oplus t_{u-\hat{z}} \oplus t_{u+\hat{z}})\equiv 0.
\]
Using the parity checks Eq.~(\ref{checksXY}) we conclude that
the restriction of $t$ onto any $xy$-plane must belong to $\calC(g,g)$.
Hence any $xy$-plane potentially contributes $2g$ to the dimension of $\calC(g,g,g)$.
We claim that $t$ is actually uniquely determined
by the restriction of $t$ onto the pair of  $xy$-planes $(i,j,k=0)$ and $(i,j,k=1)$.
Indeed, applying the parity checks $t_{u+\hat{z}}=t_{u-\hat{z}}\oplus t_{u-\hat{x}} \oplus t_{u+\hat{x}}$,
or the equivalent parity checks $t_{u+\hat{z}}=t_{u-\hat{z}}\oplus t_{u-\hat{y}} \oplus t_{u+\hat{y}}$,
one can uniquely extend $t$ to  the remaining $xy$-planes $k=2,3,\ldots,2g-1$. Moreover, this
extension automatically satisfies the parity checks Eq.~(\ref{checksXY}) in every
$xy$-plane $k=2,3,\ldots,2g-1$.
Indeed, consider any pair of sites $v=u+\hat{z}$. Suppose $u$ belongs to the $k$-th
$xy$-plane and we have already checked that the parity checks Eq.~(\ref{checksXY})
are satisfied in all $xy$-planes $0,1,\ldots,k$. Let us check that
\be
\label{checkit}
t_{v-\hat{x}} \oplus t_{v+\hat{x}} \oplus t_{v-\hat{y}} \oplus t_{v+\hat{y}}=0.
\ee
Indeed, applying  Eq.~(\ref{checksYZ}) to the sites $u\pm \hat{x}$ we get
\be
\label{aux1}
t_{v\pm \hat{x}} = t_{u+\hat{y} \pm \hat{x}} \oplus t_{u-\hat{y} \pm \hat{x}} \oplus t_{u-\hat{z} \pm \hat{x}}.
\ee
Similarly, applying Eq.~(\ref{checksXZ}) to the sites $u\pm \hat{y}$ we get
\be
\label{aux2}
t_{v\pm \hat{y}} = t_{u+\hat{x} \pm \hat{y}} \oplus t_{u-\hat{x} \pm \hat{y}} \oplus t_{u-\hat{z} \pm \hat{y}}.
\ee
Adding together Eqs.~(\ref{aux1},\ref{aux2}) we arrive at
\[
t_{v-\hat{x}} \oplus t_{v+\hat{x}} \oplus t_{v-\hat{y}} \oplus t_{v+\hat{y}}=
t_{u-\hat{z}-\hat{x}} \oplus t_{u-\hat{z}+\hat{x}} \oplus t_{u-\hat{z}-\hat{y}} \oplus t_{u-\hat{z}+\hat{y}}=0
\]
since we assumed that the parity checks Eq.~(\ref{checksXY})
are satisfied in the $(k-1)$-th $xy$-plane.

To summarize, we have shown that $\calC(g,g,g)$ is isomorphic to the direct sum
of two copies of $\calC(g,g)$ associated with the pair of $xy$-planes $(i,j,k=0)$ and
$(i,j,k=1)$. Hence Eq.~(\ref{ggg1}) follows from Eq.~(\ref{gg}).

\subsection{Proving that $-I\notin \calS$ for the case $g>1$}
\label{subs:-I}

It remains to prove that $-I\notin \calS$.
Consider any string $t\in \calC$ such  that
\be
\label{S(t)}
S(t)\equiv\prod_{u\in \Lambda_{odd}} S_u^{t_u}=\epsilon(t) \, I, \quad \epsilon(t)=\pm 1.
\ee
Since the generators $S_u$ pairwise commute and $S_u^2=I$ we conclude that
$\epsilon(t)$ is a homomorphism from $\calC$ to $\ZZ_2$, that is,
$\epsilon(t\oplus t')=\epsilon(t)\epsilon(t')$ for any $t,t'\in \calC$. Thus it suffices to check that $\epsilon(t)=+1$
only for basis vectors $t\in \calC$.
To this end we define a two-dimensional version of the generators in Eq.~(\ref{generator}),
namely,
\[
S_u'=X_{u-\hat{x}} X_{u+\hat{x}} Y_{u-\hat{y}} Y_{u+\hat{y}},
\]
where  $u\in \Lambda_{odd}$. Note that $S_u'$ is the restriction of
the 3D generator $S_u$  onto the $xy$-plane that contains the site $u$.
Define also an operator
\be
\label{S'(t)}
S'(t):=\prod_{u\in \Lambda_{odd}} (S_u')^{t_u}.
\ee
For any $xy$-plane $\Omega$ a generator $S_u$ with $u\notin \Omega$ acts on
qubits of $\Omega$ only by $I$ or $Z$. Hence the cancelation in Eq.~(\ref{S(t)}) is possible
only if $S'(t)$ is a Pauli operator of $Z$-type. More specifically, one must have
\[
S'(t)=\epsilon(t)\, \prod_{v\in \Lambda_{even}}\, Z_v^{s_v}
\]
for some string $s\in \calB(\Lambda_{even})$. Since each $Z_v$
can be only obtained from the product $X_v Y_v$ or $Y_v X_v$, we conclude that
\be
\label{st}
s_v=t_{v-\hat{x}}\oplus t_{v+\hat{x}} = t_{v-\hat{y}}\oplus t_{v+\hat{y}}.
\ee
To determine the phase factor $\epsilon(t)$
let us choose a specific
ordering of generators in Eq.~(\ref{S'(t)})
such that generators $S'_u$ with odd $x$-coordinate  are always on the left
of generators with even $x$-coordinate. (Since the generators $S'_u$ pairwise commute, the order
does not matter.) Then any qubit $v$ with even $x$-coordinate can only produce
$Z_v$ through a multiplication $X_v Y_v$. Any qubit $v$ with odd $x$-coordinate can only produce
$Z_v$ through a multiplication $Y_v X_v$. Hence we arrive at
\[
\epsilon(t)=i^{n_{even} -n_{odd}},
\]
where $n_{even}$ and $n_{odd}$ is the number of sites $v\in \Lambda_{even}$ with
$s_v=1$ such that  $v$ has even and odd $x$-coordinate respectively.
If we sum Eq.~(\ref{st}) over all sites $v$ in $\Lambda_{even}$ with odd $x$-coordinate, we see that $n_{odd}$ must be even.
Hence we arrive at
\be
\label{phase}
\epsilon(t)=i^{n_{even}+n_{odd}}=i^{|s|},
\ee
where $|s|$ is the Hamming weight of $s$.
Using the identity $|a\oplus b|=a+b-2ab$ which holds for any $a,b\in \{0,1\}$ and Eq.~(\ref{st})
one can easily get
\bea
|s| &=&\sum_{u\in \Lambda_{110}\cup \Lambda_{101}}\, |t_{u-\hat{x}} \oplus t_{u+\hat{x}}| +
\sum_{u\in \Lambda_{000}\cup \Lambda_{011}}\, |t_{u-\hat{y}} \oplus t_{u+\hat{y}}| \nn \\
&=& \sum_{u\in \Lambda_{001}\cup \Lambda_{010}}\,
4|t_u| -2 \sum_{u\in \Lambda_{001}\cup \Lambda_{010}}\, t_u t_{u+2\hat{x}} + t_u t_{u+2\hat{y}}.\nn
\eea
Now it follows from Eq.~(\ref{phase}) that
\[
\epsilon(t)=(-1)^{f(t)}, \quad f(t)=\sum_{u\in \Lambda_{001}\cup \Lambda_{010}}\, t_u t_{u+2\hat{x}} + t_u t_{u+2\hat{y}} \modtwo.
\]
Since $\epsilon$ is a homomorphism,
it suffices to check that $f(t)=0$  for all basis vectors of $\calC$.
We shall use  the isomorphism between $\calC(p_x,p_y,p_z)$ and $\calC(g,g,g)$
described in the beginning of Section~\ref{subs:k}.
Let $t'\in \calC(g,g,g)$ be the image of $t$ under this isomorphism.
Using Eq.~(\ref{period2g}) one can easily get
\[
f(t)=\frac{p_xp_yp_z}{g^3} f(t').
\]
Here $f(t')$ is calculated for generators $S_u$ defined on the lattice
$\Lambda'=\ZZ_{2g}\times \ZZ_{2g} \times \ZZ_{2g}$.
Hence it suffices to check that $f(t')=0$ for all basis vectors $t'\in \calC(g,g,g)$.
As was mentioned in Section~\ref{subs:k}, the restriction of $\calC(g,g,g)$
onto any $xy$-plane $\Omega\cong \ZZ_{2g}\times \ZZ_{2g}$ coincides with the subspace $\calC(g,g)$,
see Eq.~(\ref{2dCode}).
Let $r\in \calC(g,g)$ be the basis vector defined on Fig.~\ref{fig:CA}.
Since $r$ is symmetric under $90^\circ$ rotations of the lattice, one can easily get
\[
\sum_{u=(2i,j)\in \Omega_{odd}}\, r_u r_{u+2\hat{x}} + r_u r_{u+2\hat{y}}=0 \modtwo.
\]
Hence any $xy$-plane yields even contribution to $f(t')$, that is, $f(t')$ is even.
Thus $f(t)$ is even for all $t\in \calC(p_x,p_y,p_z)$ which
implies  $\epsilon(t)=1$.
 We have proved that $-I\notin \calS$ for $g>1$.

\section{Proof of zero-temperature stability}
\label{sec:stability}
In this Appendix we explicitly prove that the Hamiltonian  Eq.~(\ref{H}) defines a stable phase of matter at zero temperature, that is, the spectral gap of the Hamiltonian does not close in the presence of weak local perturbations and its ground state degeneracy is not lifted up to exponentially small corrections. Specifically, we will check the sufficient conditions for stability derived in~\cite{BHM:stab}.
We also show that the model Eq.~(\ref{H})  has a unique ground state if defined on an infinite lattice.

The following lemma proves two statements about Pauli operators $P$ whose support can be bounded by a sufficiently small box $B$. First, if such an operator $P$ commutes with
all generators $S_u$ then $P$ must be in the stabilizer group (up to a phase factor). Secondly,
if $P$ is a stabilizer then one can express $P$ as a product of generators
$S_u$ using only those generators supported inside the box $B$.  The first statement implies that the
stabilizer
code $\calS=\la S_u, \; u\in \Lambda_{odd}\ra$ has a macroscopic distance (growing at least linearly with the smallest of the lattice dimensions). It provides the first stability condition called TQO-1 in Ref.~\cite{BHM:stab}.
The second statement is equivalent to the second stability condition called TQO-2 in Ref.~\cite{BHM:stab}, see Lemma~2.1 in the above reference. As was shown in~\cite{BHM:stab}, conditions TQO-1,2 together are sufficient for zero-temperature stability.

\begin{lemma}
\label{lemma:stability}
Suppose a Pauli operator $P$ commutes with all generators
$S_u$, $u \in \Lambda_{odd}$ and suppose the support of $P$ can be bounded by
a box $B$ of size $l_x\times l_y\times l_z$ with
$l_\alpha \le L_\alpha -3$ for all $\alpha=x,y,z$.
Then $P$ is a stabilizer up to a phase factor.
Moreover, there exists a subset of generators $M\subseteq B\cap \Lambda_{odd}$
such that $P= \prod_{u\in M} S_u$ up to a phase factor.
\end{lemma}
\begin{proof}
Let $l=\max{(l_x,l_y,l_z)}$. We prove the lemma using induction
in $l$. The base of induction is $l=2$. Note that a box of size $2\times 2\times 2$
contains four even sites. Using the translational symmetry we can assume that $P$ acts non-trivially on at most four qubits located at sites $u_1=(0,0,0)$, $u_2=(1,1,0)$, $u_3=(1,0,1)$, and $u_4=(0,1,1)$. For any site $u_i$ one can choose a triple of generators $S_v$ acting on $u_i$ by $X$, $Y$, and $Z$, and acting trivially on the remaining three sites $u_j$, $j\ne i$. Hence the commutativity $PS_v=S_vP$ implies that $P=I$ up to  a phase factor.

Let us now prove the step of induction. Assume without loss of generality that
$l_z\ge l_x,l_y$. We can also assume that $l_z\ge 3$ since otherwise we have the base of induction. We shall construct a stabilizer $S\in \calS$ such that $PS$ has support
in a box of size $l_x\times l_y \times (l_z-1)$. This stabilizer $S$ will only use generators
whose support is contained in $B$. Indeed, let $F$ be the upper face of $B$.
The commutation $PS_v=S_vP$ implies that $P_u\in \{I,Z\}$ for all $u\in F$. In addition,
$P_u=I$ iff $u$ lies on an edge of $F$. Define the stabilizer $S$ as
\be
S=\prod_{u\in F\, : \, P_u=Z}\; S_{u-\hat{z}}.
\ee
Note that $S$ has support in $B$ and $PS$ acts trivially on $F$.
Hence $PS$ has support in a box of size  $l_x\times l_y \times (l_z-1)$.
It proves the step of induction.
\end{proof}

A simple corollary of the lemma is that the model Eq.~(\ref{H}) has unique ground state
if defined on an infinite lattice, $\Lambda=\ZZ\times \ZZ\times \ZZ$.
Indeed, define a stabilizer group $\calS^*$ generated by {\em finite} products
of generators $S_u$, $u\in \Lambda_{odd}$.
For any finite subset  $M\subseteq \Lambda_{even}$ let
$\calS(M)$ be the subgroup including all elements $S\in \calS^*$
whose support is contained in $M$. We note  that  $-I\notin \calS(M)$ since otherwise
one would have $-I\in \calS$ for a finite lattice of sufficiently large size which contradicts
to results obtained in Appendix~\ref{sec:appA}. Hence $\calS(M)$ can be regarded as a stabilizer
code by itself. Define a mixed state $\rho_M$ proportional to the projector onto the codespace
of $\calS(M)$,
\[
\rho_M=\frac1{2^{|M|}} \sum_{P\in \calS(M)}\, P.
\]
Note that for any pair of subsets $M'\subseteq M$ the inclusion
$\calS(M')\subseteq \calS(M)$ implies $\rho_{M'} =\trace_{M\backslash M'} \, \rho_M$.
Hence any pair of states $\rho_M$, $\rho_K$ with overlapping supports
have consistent marginal states on the intersection $M\cap K$.
A collection of such states $\rho=\{\rho_M\}$ associated with all finite subsets $M\subseteq \Lambda_{even}$
defines a quantum state of the entire lattice. By construction one has $S_u \rho_M=\rho_M$
for any generator $S_u$ whose support is contained in $M$. Hence $\rho$ minimizes every term in
the Hamiltonian Eq.~(\ref{H}) thus being the ground state of the model. We claim that this state is pure, namely,
$\rho$ cannot be represented as a mixture of two different quantum states.
Indeed, suppose  there exists a pair of quantum states $\tau=\{ \tau_M\}$ and $\eta=\{\eta_M\}$ such that
\be
\label{mixture}
\rho_M=(1/2)(\tau_M+\eta_M)
\ee
for all finite $M\subseteq \Lambda_{even}$. Let us show that this is possible only if $\rho_M=\tau_M=\eta_M$.
Indeed, since $\rho_M$ is proportional to the projector onto the $\calS(M)$-invariant subspace,
we conclude that the range of $\tau_M$  is spanned by $\calS(M)$-invariant states, that is,
\be
\label{range}
P\tau_M =\tau_M P = \tau_M \quad \mbox{for any $P\in \calS(M)$}.
\ee
Choose any $P\in \calS^*$ and let
$K=\mbox{Supp}(P) \cup M$. Consider the identity $P\tau_K P =\tau_K$
which follows from Eq.~(\ref{range}).
Taking a partial trace over $K\backslash M$ we get $P_M \tau_M P_M =\tau_M$,
where $P_M$ is the projection of $P$ onto $M$.  Thus $\tau_M$ commutes with
projections of stabilizer onto $M$.
Hence any Pauli operator $P$ in the Pauli expansion of $\tau_M$ commutes with
projections of stabilizer onto $M$ and thus $P\in \calS(M)$ up to a phase
factor, see Lemma~\ref{lemma:stability}. We arrive at
\[
\tau_M=\frac1{2^{|M|}} \sum_{P\in \calS(M)} a(P)\, P,
\]
for some real coefficients $a(P)$. Using Eq.~(\ref{range})
we conclude that $a(P)=a(I)$ for all $P\in \calS(M)$, that is, $\tau_M=\rho_M$.
The same argument shows that $\eta_M=\rho_M$. Thus we proved that $\tau=\eta=\rho$,
i.e., $\rho$ is a pure state.

\bibliographystyle{unsrt}

\end{document}